\documentclass[format=acmsmall, review=false]{acmart}
\pdfoutput=1
\usepackage{acm-ec-18}

\usepackage{graphicx}  
\usepackage{color}
\usepackage{mathrsfs}
\usepackage{amssymb,amsmath,verbatim}
\usepackage{amsfonts,amsthm}
\usepackage {color}
\usepackage{algorithm}
\usepackage{algorithmic}
\usepackage{booktabs,microtype}
\usepackage{tikz}
\usetikzlibrary{arrows,shapes, calc, fit, positioning}
\usepackage{tkz-graph}

\usepackage{verbatim}
\usepackage{etoolbox}

\newcounter{Bew1}
\newcounter{Bew2}

\newcommand{\bbN}{\mathbb{N}}

\newcommand{\mms}{\mathrm{mms}}
\newcommand{\calC}{\mathcal{C}}

\newcommand{\calB}{\mathcal{B}}
\newcommand{\calS}{\mathcal{S}}
\DeclareMathAlphabet{\altmathcal}{OMS}{cmsy}{m}{n} 
\newcommand{\calA}{\altmathcal{A}}
\renewcommand*{\ge}{\geqslant}
\renewcommand*{\le}{\leqslant}
\renewcommand*{\geq}{\geqslant}
\renewcommand*{\leq}{\leqslant}

\renewcommand*{\ge}{\geqslant}
\renewcommand*{\le}{\leqslant}

\numberwithin{equation}{section}

\begin{document}

\title{Pareto-Optimal Allocation of Indivisible Goods with Connectivity Constraints}

	\author{Ayumi Igarashi}
	\affiliation{Kyushu University, Fukuoka, Japan}
	\author{Dominik Peters}
	\affiliation{University of Oxford, Oxford, U.K.}

\begin{abstract}
We study the problem of allocating indivisible items to agents with additive valuations, under the additional constraint that bundles must be connected in an underlying item graph. Previous work has considered the existence and complexity of fair allocations. We study the problem of finding an allocation that is Pareto-optimal. While it is easy to find an efficient allocation when the underlying graph is a path or a star, the problem is NP-hard for many other graph topologies, even for trees of bounded pathwidth or of maximum degree 3. We show that on a path, there are instances where no Pareto-optimal allocation satisfies envy-freeness up to one good, and that it is NP-hard to decide whether such an allocation exists, even for binary valuations. We also show that, for a path, it is NP-hard to find a Pareto-optimal allocation that satisfies maximin share, but show that a moving-knife algorithm can find such an allocation when agents have binary valuations that have a non-nested interval structure.
\vspace{-10pt}
\end{abstract}

\maketitle

\section{Introduction}	\label{sec:intro}
In mechanism design with or without money, Pareto-optimality is a basic desideratum: if we select an outcome that is Pareto-dominated by another, users will justifiably complain.
In simple settings, it is computationally trivial to find a Pareto-optimum (e.g., via serial dictatorship). Thus, it is usually sought to be satisfied together with other criteria (like fairness or welfare maximisation).
However, in more complicated settings, even Pareto-optimality may be elusive.

We study the classical problem of allocating indivisible items among agents who have (typically additive) preferences over bundles. 
Following a recent model of \citet{Bouveret2017}, we are interested in settings where the set of items has additional structure specified by a graph $G$ over the items. Agents are only interested in receiving a bundle of items that is connected in $G$.
This model is particularly relevant when the items have a spatial or temporal structure, for example, if we wish to allocate land, rooms, or time slots to agents. 
Time slots, for instance, are naturally ordered in a sequence, and agents will often only value being allocated a contiguous chunk of time, particularly when restart costs are prohibitive.

Given agents' preferences over (connected) bundles, we wish to find an allocation that is \emph{Pareto-optimal} (or \emph{Pareto-efficient}), that is, a connected allocation such that there is no other \emph{connected} allocation which makes some agent strictly better off while making no agent worse off. 
Now, in the standard setting without connectivity constraints and with additive valuations, it is straightforward to find Pareto-optima: For example, we can allocate each item to a person who has the highest valuation for it (maximising utilitarian social welfare in the process), or we can run a serial dictatorship.
However, neither of these approaches respect connectivity constraints. 
In fact, we show that it is NP-hard to construct a Pareto-optimal allocation under connectivity constraints, unless $G$ is extremely simple.

Recent work on the allocation of indivisible items has focussed particularly on ensuring \emph{fairness}. 
Two well-studied fairness notions are due to \citet{Budi11a}, who introduced the \emph{maximin fair share} (MMS) and \emph{envy-freeness up to one good} (EF1).
Both concepts have natural analogues in the setting with connectivity constraints \citep{Bouveret2017,Bilo2018}.
An important question is whether there is a tradeoff between efficiency and fairness, or whether both are simultaneously achievable. In the setting without connectivity constraints, these notions tend to be compatible. 
For example, with additive valuations, \citet{CKM+16a} showed that the maximum Nash welfare solution satisfies EF1 while also being Pareto-optimal.
We investigate whether such tradeoffs exist in the connected setting.
\smallskip

\noindent
\textbf{Contributions.}
\begin{itemize}
	\item For additive valuations, we show that one can find a Pareto-optimum in polynomial time when $G$ is a path or a star. 
	\item We show that, unless P = NP, there is no polynomial-time algorithm that finds a Pareto-optimum when $G$ is a tree, even if valuations are additive and binary, and even if the tree has bounded pathwidth, bounded diameter, or bounded maximum degree. Finding a Pareto-optimum is also hard valuations are 2-additive and $G$ is a path or a star.
	\item When $G$ is a tree, there always exists an allocation which is both Pareto-optimal and satisfies MMS. However, such an allocation is NP-hard to find, even when $G$ is a path; the problem stays hard when weakening MMS to $\alpha$-MMS for any $\alpha > 0$. For a restricted class of binary valuations (non-nested intervals), we give a polynomial-time algorithm.
	\item When $G$ is a path, we give examples with binary additive valuations for which no Pareto-optimal EF1 allocation exists, and show that it is NP-hard to decide whether such an allocation exists.
\end{itemize}

\renewcommand{\arraystretch}{1.2}
\begin{table}[t]
	\tabcolsep=1.5mm
	\centering
	\small
	\begin{tabular}{lllll}
		\toprule
		& general & complete & tree & path\\
		\midrule
		PO  & NP-hard* & poly-time & NP-hard*  & poly-time
		\\
		PO \& MMS & NP-hard* &  & NP-hard* & NP-hard*
		\\
		PO \& EF1 & NP-hard & poly-time$^\dagger$ & NP-hard & NP-hard 
		\\
		\bottomrule
	\end{tabular}
\caption{Overview of our complexity results. Hardness results marked $^*$ hold under Turing reductions. 
	The result $\dagger$ refers to the pseudo-polynomial algorithm by \citet{Barman2018}. 
	The hardness results hold even for additive and binary valuations.}
\label{table}
\end{table}

{\bf Related Work} There is a rich body of the literature on fair division of a {\em divisible} cake into connected pieces. Such a division satisfying envy-freeness always exists \citep{Stromquist80}; nevertheless, it cannot be obtained in finite steps even when the cake is divided among three agents \citep{Stromquist}. In contrast, several efficient algorithms are known to yield a contiguous proportional allocation; see the survey by \citet{Lindner2016} for more details. 

The relation between efficiency and fairness with connected pieces is also well-understood for divisible items. \citet{AumannDomb2010} studied the efficiency loss of fair allocations under connectivity constraints. The papers by \citet{Bei12} and \citet{AumannDH13} considered the computational complexity of finding an allocation with connected pieces maximising utilitarian social welfare.
\citet{Bei12} showed that utilitarian social welfare is inapproximable when requiring that the allocation satisfy proportionality; however, without the proportionality requirement, \citet{AumannDH13} proved that there is a polynomial-time constant-factor approximation algorithm for finding an allocation maximising utilitarian social welfare. The algorithm in \citep{AumannDH13} works also for indivisible items and so applies to our setting when $G$ is a path.
A paper by \citet{Conitzer2004} considers combinatorial auctions; translated to our setting, their results imply that one can find a Pareto-optimal connected allocation in polynomial time, when $G$ is a graph of bounded treewidth and agents have \emph{unit demand}: each agent has a connected demanded bundle such that agents have positive utility if and only if they obtain a superset of the demanded bundle.

In the context of division of indivisible items, \citet{Bouveret2017} formalized the model of fair division with the extra feature that each bundle needs to be connected in the underlying item graph. While they showed that finding an envy-free or proportional connected allocation is NP-hard even on paths, they proved that an allocation satisfying the maximin fair share always exists and can be found in polynomial time when the graph is acyclic; subsequently \citet{Lonc2018} studied the computational complexity of finding an MMS allocation when the graph $G$ is a cycle. Independently of \citet{Bouveret2017}, \citet{Suksompong2017} considered the problem when the items lie on a path, obtaining approximations to several fairness notions such as envy-freeness and proportionality. The recent works of \citet{Bilo2018} and \citet{Oh2018} study the existence of EF1 allocations with connected pieces. They both showed that an EF1 allocation exists when agents have identical valuations. \citet{Bilo2018} also proved that for up to four agents with arbitrary monotonic valuations, an EF1 allocation connected on a path is guaranteed to exist.
Finally, \citet{ABL+16a} studied the computational complexity of finding Pareto-improvements of a given allocation when agents have additive preferences, in the setting without connectivity constraints.
Technically, our hardness proofs use similar techniques to hardness proofs obtained by \citet{AzizBH13} in the context of hedonic games.

\section{Preliminaries}\label{sec:prem}
For an integer $s \in \bbN$, write $[s]=\{1,2,\ldots,s\}$. 
Let $N=[n]$ be a set of \emph{agents} and $G=(V,E)$ be an undirected graph whose vertices are called \emph{items}.
A subset $X$ of $V$ is \emph{connected} if it induces a connected subgraph of $G$.
We write $\calC(V) \subseteq 2^V$ for the set of connected subsets of $V$, also called \emph{bundles}.

Each agent $i \in N$ has a \emph{valuation function} $u_i: \calC(V) \rightarrow \mathbb R$ over connected bundles which satisfies $u_i(\emptyset)=0$ and is \emph{monotonic}, so $X \subseteq Y$ implies $u_i(X) \leq u_i(Y)$.
%
A valuation function $u_i$ is \emph{additive} if $u_i(X)=\sum_{v \in X}u_i(\{v\})$ for each $X \in \calC(V)$. We write $u_i(v) = u_i(\{v\})$ for short. An additive valuation function is \emph{binary} if $u_i(v) \in \{0,1\}$ for all $v\in V$. If an agent $i$ has a binary valuation function, we say that $i$ \emph{approves} item $v$ if $u_i(v) = 1$.

A (connected) \emph{allocation} is a function $\pi:N \rightarrow \calC(V)$ assigning each agent a connected bundle of items, such that each item is allocated exactly once, i.e., $\bigcup_{i \in N}\pi(i)=V$ and $\pi(i) \cap \pi(j) =\emptyset$ for each pair of distinct agents $i,j \in N$. For an allocation $\pi$ and a subset $N'$ of agents, we denote by $\pi|_{N'}$ the allocation restricted to $N'$. 

Given an allocation $\pi$, another allocation $\pi'$ is a {\em Pareto-improvement} of $\pi$ if $u_i(\pi'(i)) \ge u_i(\pi(i))$ for all $i \in N$ and $u_j(\pi'(j)) >  u_j(\pi(j))$ for some $j \in N$. 
We say that a connected allocation $\pi$ is {\em Pareto-optimal} (or \emph{Pareto-efficient}, or \textit{PO} for short) if there is no \emph{connected} allocation that is a Pareto-improvement of $\pi$. 
The \emph{utilitarian social welfare} of an allocation $\pi$ is $\sum_{i\in N} u_i(\pi(i))$. 
It is easy to see that a connected allocation which maximizes utilitarian social welfare among connected allocations is Pareto-optimal.

A connected allocation satisfies \emph{EF1} \citep{Bilo2018,Oh2018} if for any pair of agents $i,j \in N$, either $u_i(\pi(i)) \ge u_i(\pi(j))$ or there is an item $v \in \pi(j)$ such that $\pi(j)\setminus \{v\}$ remains connected and $u_i(\pi(i)) \ge u_i(\pi(j)\setminus \{v\})$. Thus, whenever $i$ envies the bundle of agent $j$, then the envy vanishes if we remove one outer item from the envied bundle.

Let $\Pi_n(G)$ denote the set of partitions of $V$ into $n$ connected bundles of the graph $G$.
The {\em maximin fair share} of an agent $i\in N$
is 
\[
\mms_i = \max_{(P_1, \dots, P_n)\in\Pi_n(G)}\min_{j\in [n]} u_i(P_j). 
\]
A connected allocation $\pi$ is an \emph{MMS allocation} if $u_i(\pi(i)) \ge \mms_i$ for each agent $i \in N$. \citet{Bouveret2017} show that if $G$ is a tree, an MMS allocation exists. Note that this definition of the MMS value varies with the graph $G$, and may be lower than the standard MMS values where the maximisation is taken over all partitions, with no connectivity constraints.

Finally we introduce some graph-theoretic terminology. 
Given a graph $G=(V,E)$ and a subset $X \subseteq V$ of vertices, we denote by $G\setminus X$ the subgraph of $G$ induced by $V\setminus X$. 
The \emph{diameter} of $G$ is the maximum distance between any pair of vertices. For two paths $P_1=(a_1,a_2,\ldots,a_s)$ and $P_2=(b_1,b_2,\ldots,b_t)$, we define the {\em concatenation} $P_1P_2$ of $P_1$ and $P_2$ as follows: $P_1P_2:=(a_1,a_2,\ldots,a_s,b_1,b_2,\ldots,b_t)$, where there is a new edge between $a_s$ and $b_1$. The {\em concatenation} of a finite sequence of paths $P_1, P_2,\ldots,P_k$ can be defined inductively.

\section{Finding Some Pareto-Optimal Allocation}\label{sec:pareto}
We start by considering the problem of producing \textit{some} Pareto-optimal allocation, without imposing any additional constraints on the quality of this allocation. When there are no connectivity requirements (equivalently, when $G$ is a complete graph) and valuations are additive, this problem is trivial: Simply allocate each item $v$ separately to an agent $i$ who has a highest valuation $u_i(v)$ for $v$. The resulting allocation maximizes utilitarian social welfare and is thus Pareto-optimal. When $G$ is not complete, this procedure can produce disconnected bundles. 
We could try to give all items to a single agent (which satisfies our constraints provided that $G$ is connected), but the result need not be Pareto-optimal.
Is it still possible to find a Pareto-optimal allocation for specific graph topologies in polynomial time?

\subsection{Paths and Stars}
For very simple graphs and additive valuations, the answer is positive. 
Our first algorithm works when $G$ is a path. The algorithm identifies an agent $i$ with a nonzero valuation for the item at the left end of the path $G$, and then allocates all items to $i$, except for any items at the right end of the path which $i$ values at 0. We then recursively call the same algorithm to decide on how to allocate the remaining items.

\begin{theorem}\label{thm:PO:path}
When $G$ is a path, and with additive valuations, a Pareto-optimal allocation can be found in polynomial time.
\end{theorem}
\begin{proof}
The path $G$ is given by $V=\{v_1,v_2,\ldots,v_{m}\}$ where $\{v_j,v_{j+1}\} \in E$ for each $j \in [m-1]$. For a subset $V'$ of $V$, we denote by $\min V'$ the vertex of minimum index in $V'$. 

We design a recursive algorithm $\calA$ that takes as input a subset $N'$ of agents, a subpath $G'=(V',E')$ of $G$, and a valuation profile $(u_i)_{i \in N'}$, and returns a Pareto-optimal allocation of the items in $V'$ to the agents in $N'$. Without loss of generality, we may assume that there is an agent who likes the left-most vertex of $G'$, i.e., $u_i(\min V')>0$ for some $i \in N'$, since otherwise we can arbitrarily allocate that item later without affecting Pareto-optimality.

If $|N'|=1$, then the algorithm allocates all items $V'$ to the single agent. Suppose that $|N'|>1$. 
The algorithm first finds an agent $i$ who has positive value for $\min V'$; it then allocates to $i$ a minimal connected bundle $V_i \subseteq V'$ containing all items in $V'$ to which $i$ assigns positive utility (so that $u_i(V_i) = u_i(V')$). To decide on the allocation of the remaining items, we apply $\calA$ to the reduced instance $(N'\setminus \{i\},G'\setminus V_i,(u_{i'})_{{i'} \in N' \setminus \{i\}})$.  

We will prove by induction on $|N'|$ that the allocation produced by $\calA$ is Pareto-optimal. This is clearly true when $|N'|=1$. Suppose that $\calA$ returns a Pareto-optimal allocation when $|N'| = k-1$; we will prove it for $|N'|=k$. Let $\pi$ be the allocation returned by $\calA$, where $\calA$ chose agent $i$ and allocated the bundle~$V_i$ before making a recursive call. Assume for a contradiction that there is a Pareto-improvement $\pi'$ of $\pi$. Thus, in particular, $u_i(\pi'(i)) \geq u_i(\pi(i))$. Hence, by the algorithm's choice of the bundle $V_i$, we must have $V_i \subseteq \pi'(i)$ and $u_i(\pi'(i)) = u_i(\pi(i))$. Hence, there must be an agent $j'$ different from $i$ who receives strictly higher value in $\pi'$ than in $\pi$. 

Now, since $G \setminus \pi'(i)$ is a subgraph of $G \setminus V_i$, the allocation $\pi'|_{N' \setminus \{i\}}$ is an allocation for the instance $I' = (N'\setminus \{i\},G'\setminus V_i,(u_{i'})_{{i'} \in N' \setminus \{i\}})$. Also, we have 
\begin{itemize}
\item $u_{j}(\pi'(j)) \geq u_{j}(\pi(j))$ for all agents $j \in N' \setminus \{i\}$; and
\item $u_{j'}(\pi'(j'))> u_{j'}(\pi(j'))$ for some $j' \in N' \setminus \{i\}$.
\end{itemize}
Thus, $\pi'|_{N' \setminus \{i\}}$ is a Pareto-improvement of the allocation $\pi|_{N' \setminus \{i\}}$. But $\pi|_{N' \setminus \{i\}}$ is the allocation returned by $\calA$ for the instance $I'$, contradicting the inductive hypothesis that $\calA$ returns Pareto-optimal allocations for $|N'| = k-1$. 
\end{proof}

Another graph type for which we can find a Pareto-optimum is a star. In fact, we can efficiently calculate an allocation maximizing utilitarian social welfare. Note that when $G$ is a star, at most one agent can receive two or more items. This allows us to translate welfare maximization into a bipartite matching instance.

\begin{theorem}
\label{thm:PO:star}
When $G$ is a star, and valuations are additive, a Pareto-optimal allocation can be found in polynomial time.
\end{theorem}
\begin{proof}
We give an algorithm to find an allocation that maximizes the utilitarian social welfare. 
Let $G$ be a star with center vertex $c$ and $m-1$ leaves.
We start by guessing an agent $i \in N$ who receives the item $c$. 
By connectedness, every other agent can receive at most one (leaf) item. 
To allocate the leaf items, we construct a weighted bipartite graph $H_i$ with bipartition $(N',V\setminus \{c\})$ where $N'$ consists of agents $j \in N\setminus\{i\}$ together with $m-1$ copies $i_1,i_2,\ldots,i_{m-1}$ of agent $i$. 
(These copies represent `slots' in $i$'s bundle.) 
Each edge of form $\{j,v\}$ for $j \in N\setminus\{i\}$ has weight $u_j(v)$ and each edge of form $\{i_k,v\}$ has weight $u_i(v)$.

Observe that each connected allocation in which $i$ obtains $c$ can be identified with a matching in $H_i$: Every leaf object is either matched with the agent receiving it, or is matched with some copy $i_k$ of $i$ if the object is part of $i$'s bundle. Note that utilitarian social welfare of this allocation equals the total weight of the matching. Since one can find a maximum-weight matching in a bipartite graph in polynomial time \citep[see, e.g.,][]{Korte}, we can find an allocation of maximum utilitarian social welfare efficiently. 
\end{proof}

We have shown that finding a Pareto-optimum is easy for paths and for stars. 
An interesting open problem is whether the problem is also easy for caterpillars, a class of graphs containing both paths and stars. One might be able to combine the approaches of Theorems~\ref{thm:PO:path} and~\ref{thm:PO:star} to handle them, but the details are difficult. 
Note that caterplliars are precisely the graphs of pathwidth~1; we discuss a negative result about graphs of pathwidth~2 below. 
Another open problem is whether finding a Pareto-optimum is easy when $G$ is a cycle.

\subsection{Hardness Results}\label{sec:pareto-hard}
For more general classes of graphs, the news is less positive. We show that, unless P = NP, there is no polynomial-time algorithm which produces a Pareto-optimal allocation when $G$ is an arbitrary tree even for binary valuations. Notably, this result implies that it is NP-hard to find allocations maximizing any type of social welfare (utilitarian, leximin, Nash) when $G$ is a tree, since such allocations are also Pareto-optimal.

To obtain our hardness result, we first consider a more general problem which is easier to handle, namely the case where $G$ is a forest. 
Since a Pareto-optimum always exists, we cannot employ the standard approach of showing that a decision problem is NP-hard via many-one reductions; instead we will show (using a Turing reduction) that a polynomial-time algorithm producing a connected Pareto-optimal allocation could be used to solve an NP-complete problem in polynomial time.

\begin{theorem}\label{thm:forest:Perfect}
	Unless \textup{P = NP}, there is no polynomial-time algorithm which finds a Pareto-optimal connected allocation when $G$ is a union of vertex-disjoint paths of size 3, even if valuations are binary and additive.
\end{theorem}
\begin{proof}
We give a Turing reduction from {\sc Exact-3-Cover (X3C)}. Recall that an instance of {\sc X3C} is given by a set of elements $X=\{x_1,x_2,\dots, x_{3r}\}$ and a family $\calS = \{S_1,\dots, S_s\}$ of three-element subsets of $X$;
it is a `yes'-instance if and only if there is an {\em exact cover} $\calS' \subseteq \calS$ with $|\calS'|=r$ and $\bigcup_{S \in \calS'}S = X$. For a set $S\in \calS$, order the three elements of $S$ in some canonical way (e.g., alphabetically) and write $S^1, S^2, S^3$ for the elements in that order.

Given an instance $(X, \calS)$ of X3C, for each $S\in\calS$, construct a path $P_S$ on three vertices $v_{S,1}, v_{S,2}, v_{S,3}$. Let $G = \bigcup_{S \in \calS} P_S$. 
For each element $x \in X$, we introduce an agent $i_x$ with $u_{i_x}(v_{S,j}) = 1$ iff $S^j = x$, and $u_{i_x}(v_{S,j}) = 0$ otherwise. 
Thus, agent $i_x$ approves all instances of $x$.
We also introduce $s-r$ dummy agents $d_1,\dots,d_{s-r}$ who approve every item, so $u_{d_k}(v_{S,j}) = 1$ for all $j,k,S$.
Note that for each agent $i_x$, a highest-value connected bundle has value 1, while a highest-value bundle for a dummy agent $d_k$ a highest-value connected bundle has value 3. 

Suppose we had an algorithm $\calA$ which finds a Pareto-optimal allocation. We show how to use $\calA$ to solve X3C.
Run $\calA$ on the allocation problem constructed above to obtain a Pareto-optimal allocation $\pi$.
We claim that the X3C instance $(X,\calS)$ has a solution iff
\begin{equation}
\label{eq:path-perfect}
\begin{aligned}
u_{i_x}(\pi(i_x)) &= 1 \text{ for all $x\in X$ and } \\
u_{d_k}(\pi(d_k)) &= 3 \text{ for all $k \in [s-r]$}.
\end{aligned}
\end{equation}
Since \eqref{eq:path-perfect} is easy to check, this equivalence implies that $\calA$ can be used to solve X3C, and hence our problem is NP-hard.

Suppose $\pi$ satisfies condition \eqref{eq:path-perfect}. We then construct a solution to the instance of X3C. For each $k \in [s-r]$, since $u_{d_k}(\pi(d_k)) = 3$, we must have $\pi(d_k) = P_S$ for some $S\in\calS$. Let $\calS' = \{ S \in \calS : \pi(d_k) \neq P_S \text{ for all $k \in [s-r]$}  \}$. Then $\calS'$ is a solution: Clearly $|\calS'| = r$; further, for every $x \in X$, we have that $\pi(i_x) \in P_S$ for some $S\in \calS$, and since $u_{i_x}(\pi(i_x)) = 1$ by \eqref{eq:path-perfect}, this implies that $x\in S$. Thus, $\bigcup_{S \in \calS'} S = X$. Hence, $\calS'$ is a solution to the X3C instance $(X,\calS)$.

Conversely, suppose there is a solution $\calS'$ to the instance of X3C, but suppose for a contradiction that $\pi$ fails condition \eqref{eq:path-perfect}. Define the following allocation $\pi^*$: For each $x\in X$, identify a set $S\in \calS'$ and an index $j\in [3]$ such that $S^j = x$ and set $\pi^*(i_x) = \{ v_{S,j} \}$; next, write $\calS \setminus \calS' = \{ S'_{1}, \dots, S'_{s-r}  \}$ and set $\pi^*(d_k) = \{v_{S'_k,1}, v_{S'_k,2}, v_{S'_k,3}\}$ for each $k \in [s-r]$. Then $\pi^*$ satisfies  \eqref{eq:path-perfect}. Since $\pi$ fails \eqref{eq:path-perfect}, the allocation $\pi^*$ Pareto-dominates $\pi$, contradicting that $\pi$ is Pareto-optimal. Hence, $\pi$ satisfies  \eqref{eq:path-perfect}, as desired.
\end{proof}


From this result, it immediately follows (for forests) that it is (weakly) NP-hard under Turing reductions to decide whether a given allocation is Pareto-optimal: if we could decide this in polynomial time, we could use such an algorithm to construct a Pareto-improvement, and by using this algorithm repeatedly, we could find a Pareto-optimum. 

Building on the above reduction, we find that it is also hard to find a Pareto-efficient allocation when $G$ is a tree.

\begin{figure*}[ht]
\centering
\begin{tabular}{cc}
\begin{minipage}{0.55\textwidth}
\begin{tikzpicture}[scale=0.7, transform shape, every node/.style={minimum size=8mm, inner sep=1pt}]
	\node at (0,0) {$\dots$};
	\node[draw, circle](0) at (1.2,0) {$v_{S_i,1}$};
	\node[draw, circle](1) at (2.4,0) {$v_{S_i,2}$};
	\node[draw, circle](2) at (3.6,0) {$v_{S_i,3}$};
	\node at (4.8,0) {$\dots$};
	\node at (6.2,0) {$\dots$};
	\node[draw, circle](3) at (7.4,0) {$v_{S_j,1}$};
	\node[draw, circle](4) at (8.6,0) {$v_{S_j,2}$};
	\node[draw, circle](5) at (9.8,0) {$v_{S_j,3}$};
	\node at (11,0) {$\dots$};

	\draw[-, >=latex,thick] (0)--(1);
	\draw[-, >=latex,thick] (2)--(1);
	
	\draw[-, >=latex,thick] (3)--(4);
	\draw[-, >=latex,thick] (4)--(5);
	
\end{tikzpicture}
\end{minipage}
\qquad\quad
\begin{minipage}{0.3\textwidth}
\begin{tikzpicture}[scale=0.6, transform shape, every node/.style={minimum size=8mm, inner sep=1pt}]
	\def \radius {1.5cm}
	\node[draw, circle](c) at (5.8,0) {$c$};
	\node[draw, circle](0) at (3.2,1.7) {${v}_{S_i,1}$};
	\node[draw, circle](1) at (2.8,0.7) {${v}_{S_i,3}$};
	\node[draw, circle](2) at (4,0.7) {${v}_{S_i,2}$};
	
	\node[draw, circle](5) at (7.8,1.7) {${v}_{S_j,1}$};
	\node[draw, circle](3) at (7.4,0.7) {${v}_{S_j,2}$};
	\node[draw, circle](4) at (8.6,0.7) {${v}_{S_j,3}$};
	
	\node[draw, circle](00) at (3.2,-1.7) {${\hat v}_{S_i,1}$};
	\node[draw, circle](01) at (2.8,-0.7) {${\hat v}_{S_i,3}$};
	\node[draw, circle](02) at (4,-0.7) {${\hat v}_{S_i,2}$};
	
	\node[draw, circle](03) at (7.4,-0.7) {${\hat v}_{S_j,2}$};
	\node[draw, circle](04) at (7.8,-1.7) {${\hat v}_{S_j,1}$};
	\node[draw, circle](05) at (8.6,-0.7) {${\hat v}_{S_j,3}$};
    
 	\draw[-, >=latex,thick] (c)--(2);
	\draw[-, >=latex,thick] (c)--(3);
	 \draw[-, >=latex,thick] (c)--(02);
	\draw[-, >=latex,thick] (c)--(03);

	\draw[-, >=latex,thick] (2)--(1) (0)--(2);
	\draw[-, >=latex,thick] (00)--(02) (02)--(01);
	
	\draw[-, >=latex,thick] (3)--(4) (3)--(5);
	\draw[-, >=latex,thick] (03)--(04) (03)--(05);
	
	\node at (5.7,1) {$\dots$};
	\node at (5.7,-1) {$\dots$};
\end{tikzpicture}
\end{minipage}
\end{tabular}
\caption{
Graphs constructed in the proofs of Theorem~\ref{thm:forest:Perfect} and Theorem~\ref{thm:PO:hardness:tree} (pictured left-to-right). 
\label{fig1}
}
\end{figure*}
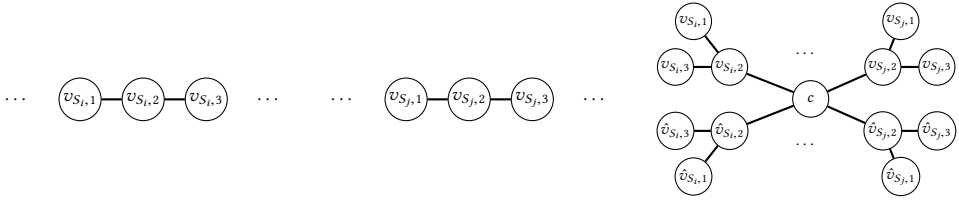

\begin{theorem}\label{thm:PO:hardness:tree}
	Unless \textup{P = NP}, there is no polynomial-time algorithm which finds a Pareto-optimal connected allocation when $G$ is a tree, even if valuations are binary and additive.
\end{theorem}
\begin{proof}
To extend the reduction in the proof of Theorem \ref{thm:forest:Perfect} to trees, we first `double' the reduction, in making a copy of each object and a copy of each agent with the same preference as the original agent. Specifically, given an instance $(X, \calS)$ of X3C, we create the same instance as in the proof of Theorem \ref{thm:forest:Perfect}; that is, we make a path $P_S =({v}_{S,1}, v_{S,2}, {v}_{S,3})$ for each $S  \in  \calS$, and construct agent $i_x$ for each $x \in X$ and dummy agents ${d}_1,d_2,\dots,{d}_{s-r}$ with the same binary valuations. 

In addition, we make a path ${\hat P}_S$ of copies ${\hat v}_{S,1}, {\hat v}_{S,2}, {\hat v}_{S,3}$ of each $S \in \calS$. 
We then make a copy $\smash{\hat i}_x$ of each agent $i_x$ $(x \in X)$ together with copies ${\hat d}_1,{\hat d}_2,\dots,{\hat d}_{s-r}$ of the dummy agents. 
We also introduce a new dummy item $c$ which serves as the center of a tree; specifically, we attach the center to the middle vertex $v_{S,2}$ of the path $P_S$, and the middle vertex ${\hat v}_{S,2}$ of the path ${\hat P}_S$, for each $S\in \calS$.
The resulting graph $G$ is a tree consisting of $2r+2|\calS|$ paths of length 3, each attached to the vertex $c$ by their middle vertex. See Figure \ref{fig1}. 

No agent has positive value for the center dummy item. Copied agents only value copied objects and have the same valuations as the corresponding original agents, and non-copied agents only value non-copied objects. 
Formally, for each element $x \in X$, each $k \in [s-r]$, and each item $v$, agents' binary valuations are given as follows: 
\begin{itemize}
\item $u_{i_x}(v) = 1$ iff $v=v_{S,j}$ and $S^j = x$;
\item $u_{{d}_k}(v) = 1$ iff $v=v_{S,j}$ for some $S,j$;
\item $u_{{\hat i}_x}(v) = 1$ iff $v={\hat v}_{S,j}$ and $S^j = x$; 
\item $u_{{\hat d}_k}(v) = 1$ iff $v={\hat v}_{S,j}$ for some $S,j$.
\end{itemize}
Write $N_{o}=\{\, i_x :  x \in X\,\} \cup \{d_1,d_2,\ldots,d_{s-r}\}$ for the set of original agents, and $V_o = \bigcup_{S \in \calS} \{ v_{S,1}, v_{S,2}, v_{S,3} \}$ for the set of original items.

Suppose we had an algorithm $\calA$ which finds a Pareto-optimal allocation. We show how to use $\calA$ to solve X3C.
Run $\calA$ on the allocation problem constructed above to obtain a Pareto-optimum $\pi$. We may suppose without loss of generality that $c \not\in \pi(i)$ for any $i \in N_{o}$, since otherwise we can swap the roles of the originals and the copies. We may further assume that each original agent $i\in N_o$ only receives original items, i.e., $\pi(i) \subseteq V_o$, since we can move any other items from $\pi(i)$ into other bundles without making anyone worse off. Hence, since $c\not\in\pi(i)$, we see that $\pi(i) \subseteq P_S$ for some $S\in\calS$ because $\pi(i)$ is connected in $G$. This implies that $u_{i_x}(\pi(i_x)) \le 1$ for all $x\in X$ and $u_{d_k}(\pi(d_k)) \le 3$ for all $k \in [s-r]$.

We now prove that the X3C instance has a solution iff
\begin{equation}
\label{eq:star-perfect}
\begin{aligned}
u_{i_x}(\pi(i_x)) &= 1 \text{ for all $x\in X$ and } \\
u_{d_k}(\pi(d_k)) &= 3 \text{ for all $k \in [s-r]$}.
\end{aligned}
\end{equation}
Since \eqref{eq:star-perfect} is easy to check, this equivalence implies that $\calA$ can be used to solve X3C, and hence our problem is NP-hard.

If \eqref{eq:star-perfect} holds, then the argument in the proof of Theorem~\ref{thm:forest:Perfect} applies and shows that the X3C instance has a solution.

Conversely, suppose there is a solution $\calS' \subseteq \calS$ to the X3C instance.
Then, as in the proof of Theorem \ref{thm:forest:Perfect}, there is an allocation $\pi^* : N_o \to \calC(V_o)$ of the original items to the original agents such that 
$u_{i_x}(\pi^*(i_x)) = 1$ for all $x\in X$ and
$u_{d_k}(\pi^*(d_k)) = 3$ for all $k \in [s-r]$.
Extend $\pi^*$ to all agents by defining $\pi^*(\hat j) = \pi(\hat j) \cap (V \setminus V_o)$ for every copied agent~$\hat j$. It is easy to check that $\pi^*$ is a connected allocation. 
For each copied agent $\hat j$, we have $u_{\hat j}(\pi^*(\hat j)) = u_{\hat j}(\pi(\hat j))$, since $\smash{\hat j}$ has a valuation of 0 for every item in $V_o$. 
Also, for each original agent $i\in N_o$, we have $u_i(\pi^*(i)) \ge u_i(\pi(i))$, since $i$ obtains an optimal bundle under $\pi^*$.
It follows that if $\pi$ fails \eqref{eq:star-perfect}, then $\pi^*$ is a Pareto-improvement of $\pi$, contradicting that $\pi$ is Pareto-optimal. So $\pi$ satisfies \eqref{eq:star-perfect}.
\end{proof}

We note that the graph constructed in the above proof has pathwidth~2 and diameter~5, so hardness holds even for trees of bounded pathwidth and bounded diameter. 
One can adapt our reduction to show that hardness also holds on trees with maximum degree~3, by copying our original reduction many times. See Figure~\ref{fig:max-degree} for the structure of the resulting graph.

\begin{theorem}\label{thm:PO:hardness:max-degree}
	Unless \textup{P = NP}, there is no polynomial-time algorithm which finds a Pareto-optimal connected allocation when $G$ is a tree with maximum degree~3, even if valuations are binary and additive.
\end{theorem}
\begin{proof}
	We give a Turing reduction from X3C similar to Theorem~\ref{thm:PO:hardness:tree}.
	
	Given an instance $(X, \calS)$ of X3C, with $X=\{x_1,x_2,\dots, x_{3r}\}$ and $\calS = \{S_1,\dots, S_s\}$, we build the following instance.
	
	\begin{figure*}[htb]
		\centering
		\begin{tikzpicture}
		[scale=0.45, transform shape, every node/.style={minimum size=14mm, inner sep=1pt, font=\huge}]
		\node[draw, circle](0) at (2,0) {$v_{S_1,1}^1$};
		\node[draw, circle](1) at (4,0) {$v_{S_1,2}^1$};
		\node[draw, circle](2) at (6,0) {$v_{S_1,3}^1$};
		\node[draw, circle](c1) at (8,0) {$c_{S_1}^1$};
		\draw[-] (0)--(1) (1)--(2) (2)--(c1);
		
		\node[draw, circle](0) at (2,2) {$v_{S_1,1}^2$};
		\node[draw, circle](1) at (4,2) {$v_{S_1,2}^2$};
		\node[draw, circle](2) at (6,2) {$v_{S_1,3}^2$};
		\node[draw, circle](c2) at (8,2) {$c_{S_1}^2$};
		\draw[-] (0)--(1) (1)--(2) (2)--(c2);
		
		\node[draw, circle](0) at (2,6) {$v_{S_1,1}^{s+1}$};
		\node[draw, circle](1) at (4,6) {$v_{S_1,2}^{s+1}$};
		\node[draw, circle](2) at (6,6) {$v_{S_1,3}^{s+1}$};
		\node[draw, circle](cs) at (8,6) {$c_{S_1}^{s+1}$};
		\draw[-] (0)--(1) (1)--(2) (2)--(cs);
		
		\node[draw, circle](b1) at (8,-2) {$b_{S_1}$};
		
		\node[font=\huge] (dots) at (8,4) {$\vdots$};
		\draw (b1)--(c1) (c1)--(c2) (c2)--(dots) (dots)--(cs);
		
		\begin{scope}[shift={(9,0)}]
		\node[draw, circle](0) at (2,0) {$v_{S_2,1}^1$};
		\node[draw, circle](1) at (4,0) {$v_{S_2,2}^1$};
		\node[draw, circle](2) at (6,0) {$v_{S_2,3}^1$};
		\node[draw, circle](c1) at (8,0) {$c_{S_2}^1$};
		\draw[-] (0)--(1) (1)--(2) (2)--(c1);
		
		\node[draw, circle](0) at (2,2) {$v_{S_2,1}^2$};
		\node[draw, circle](1) at (4,2) {$v_{S_2,2}^2$};
		\node[draw, circle](2) at (6,2) {$v_{S_2,3}^2$};
		\node[draw, circle](c2) at (8,2) {$c_{S_2}^2$};
		\draw[-] (0)--(1) (1)--(2) (2)--(c2);
		
		\node[draw, circle](0) at (2,6) {$v_{S_2,1}^{s+1}$};
		\node[draw, circle](1) at (4,6) {$v_{S_2,2}^{s+1}$};
		\node[draw, circle](2) at (6,6) {$v_{S_2,3}^{s+1}$};
		\node[draw, circle](cs) at (8,6) {$c_{S_2}^{s+1}$};
		\draw[-] (0)--(1) (1)--(2) (2)--(cs);
		
		\node[draw, circle](b2) at (8,-2) {$b_{S_2}$};
		
		\node[font=\huge] (dots) at (8,4) {$\vdots$};
		\draw (b2)--(c1) (c1)--(c2) (c2)--(dots) (dots)--(cs);
		\end{scope}
		
		\begin{scope}[shift={(20,0)}]
		\node[draw, circle](0) at (2,0) {$v_{S_s,1}^1$};
		\node[draw, circle](1) at (4,0) {$v_{S_s,2}^1$};
		\node[draw, circle](2) at (6,0) {$v_{S_s,3}^1$};
		\node[draw, circle](c1) at (8,0) {$c_{S_s}^1$};
		\draw[-] (0)--(1) (1)--(2) (2)--(c1);
		
		\node[draw, circle](0) at (2,2) {$v_{S_s,1}^2$};
		\node[draw, circle](1) at (4,2) {$v_{S_s,2}^2$};
		\node[draw, circle](2) at (6,2) {$v_{S_s,3}^2$};
		\node[draw, circle](c2) at (8,2) {$c_{S_s}^2$};
		\draw[-] (0)--(1) (1)--(2) (2)--(c2);
		
		\node[draw, circle](0) at (2,6) {$v_{S_s,1}^{s+1}$};
		\node[draw, circle](1) at (4,6) {$v_{S_s,2}^{s+1}$};
		\node[draw, circle](2) at (6,6) {$v_{S_s,3}^{s+1}$};
		\node[draw, circle](cs) at (8,6) {$c_{S_s}^{s+1}$};
		\draw[-] (0)--(1) (1)--(2) (2)--(cs);
		
		\node[draw, circle](b3) at (8,-2) {$b_{S_s}$};
		
		\node[font=\huge] (dots) at (8,4) {$\vdots$};
		\draw (b3)--(c1) (c1)--(c2) (c2)--(dots) (dots)--(cs);
		\end{scope}
		
		\node[font=\huge] (bdots) at (22,-2) {$\cdots$};
		\draw (b1)--(b2) (b2)--(bdots) (bdots)--(b3);
		
		\end{tikzpicture}
		
		\caption{
			Graph with maximum degree~3 constructed in the proof of Theorem~\ref{thm:PO:hardness:max-degree}.
			\label{fig:max-degree}
		}
	\end{figure*}
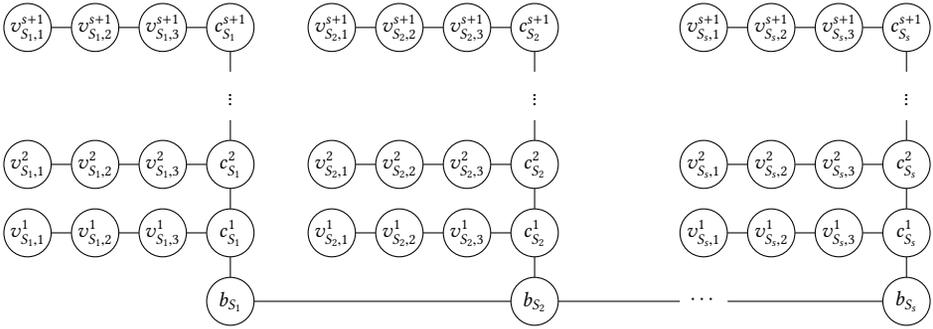
	
	Our set of agents will consist of $s+1$ copies of the set of agents in the reduction of Theorem~\ref{thm:forest:Perfect}.
	So, for each $f = 1,\dots, s+1$, we introduce the following agents:
	\[
		N^f = \{ i_x^f : x \in X  \} \cup \{ d_1^f,\dots, d_{s-r}^f \}
	\]
	The complete set of agents is then $N = N^1 \cup \cdots \cup N^{s+1}$. We call each of the sets $N^f$ a \emph{family} of agents.
	
	The set of items is $V = B \cup C \cup \bigcup_{S\in\calS} \bigcup_{f=1}^{s+1} P_S^f$, where $B = \{ b_S : S \in \calS \}$, $C = \{ c_S^f : S \in \calS, f \in [s+1]  \}$, and $P_S^f = \{ v_{S,1}^f, v_{S,2}^f, v_{S,3}^f \}$ for each $S\in \calS$ and each $f = 1,\dots,s+1$.
	
	These items are connected in a graph $G$ that is a union of the following paths: 
	\begin{itemize}
		\item The items in $B$ are connected in a path $B_{S_1}, \dots, B_{S_s}$. 
		\item For each $S\in \calS$, the items $b_S, c_S^1, \dots, c_S^{s+1}$ are connected in a path. 
		\item For each $S\in \calS$ and each $f = 1,\dots,s+1$, the items $c_S^f, v_{S,1}^f, v_{S,2}^f, v_{S,3}^f$ are connected in a path.
	\end{itemize}
	See Figure~\ref{fig:max-degree} for an illustration of this graph. Note that $G$ has maximum degree 3.
	
	We next specify the agents' binary valuations. Unless otherwise specified, each agent has valuation 0 for each item. For each $f = 1,\dots,s+1$, and for each $x \in X$, the agent $i_x^f$ approves those items $v_{S,j}^f$ such that $S^j = x$. Further, for each $f = 1,\dots,s+1$, agent $d_k^f$ approves all items $v_{S,j}^f$.
	Note that no agent approves item in $B$ or in $C$.
	
	Suppose we had an algorithm $\calA$ which finds a Pareto-optimal allocation. We show how to use $\calA$ to solve X3C.
	Run $\calA$ on the allocation problem constructed above to obtain a Pareto-optimum~$\pi$. 
	There are $s$ items in $B$, and there are $s+1$ families, so there is some family such that no item in $B$ is allocated to any family member. That is, there is some family $f$ such that $b_S \not\in \pi(i^f)$ for all $S\in\calS$ and all $i^f\in N^f$. In particular, by connectedness of $\pi$, this means that no agent in $N^f$ receives items from two different ``arms'' of $G$.
	Write $V^f = \bigcup_{S \in \calS} P_S^f$.
	We may assume that $\pi(i^f) \subseteq V^f$ for every $i^f \in N^f$, since we can move any other items from $\pi(i)$ into other bundles without making anyone worse off. Hence, we see that $\pi(i^f) \subseteq P_S^f$ for some $S\in\calS$ because $\pi(i^f)$ is connected in $G$. This implies that $u_{i_x^f}(\pi(i_x^f)) \le 1$ for all $x\in X$ and $u_{d_k^f}(\pi(d_k^f)) \le 3$ for all $k \in [s-r]$.
	
	We now prove that the X3C instance has a solution iff
	\begin{equation}
	\label{eq:max-degree-perfect}
	\begin{aligned}
	u_{i_x^f}(\pi(i_x^f)) &= 1 \text{ for all $x\in X$ and } \\
	u_{d_k^f}(\pi(d_k^f)) &= 3 \text{ for all $k \in [s-r]$}.
	\end{aligned}
	\end{equation}
	Since \eqref{eq:max-degree-perfect} is easy to check, this equivalence implies that $\calA$ can be used to solve X3C, and hence our problem is NP-hard.
	
	If \eqref{eq:max-degree-perfect} holds, then the argument in the proof of Theorem~\ref{thm:forest:Perfect} applies and shows that the X3C instance has a solution.
	
	Conversely, suppose there is a solution $\calS' \subseteq \calS$ to the X3C instance.
	Then, as in the proof of Theorem \ref{thm:forest:Perfect}, there is an allocation $\pi^* : N^f \to \calC(V^f)$ of items to the agents in family $f$ such that 
	$u_{i_x^f}(\pi^*(i_x^f)) = 1$ for all $x\in X$ and
	$u_{d_k^f}(\pi^*(d_k^f)) = 3$ for all $k \in [s-r]$.
	Extend $\pi^*$ to all agents by defining $\pi^*(j) = \pi(j) \cap (V \setminus V^f)$ for every agent~$j\in N\setminus N^f$. It is easy to check that $\pi^*$ is a connected allocation. 
	For each agent $j\in N\setminus N^f$, we have $u_{j}(\pi^*(j)) = u_{j}(\pi(j))$, since $j$ has a valuation of 0 for every item in $V^f$. 
	Also, for each agent $i\in N^f$, we have $u_i(\pi^*(i)) \ge u_i(\pi(i))$, since $i$ obtains an optimal bundle under $\pi^*$.
	It follows that if $\pi$ fails \eqref{eq:max-degree-perfect}, then $\pi^*$ is a Pareto-improvement of $\pi$, contradicting that $\pi$ is Pareto-optimal. So $\pi$ satisfies \eqref{eq:max-degree-perfect}.
\end{proof}

In the last section, we saw positive results for paths and stars when valuations are additive. 
For more general preferences over bundles, we again obtain a hardness result.

\begin{theorem}
	\label{thm:PO:hardness:path-star}
	Unless \textup{P = NP}, there is no polynomial-time algorithm which finds a Pareto-optimal connected allocation when $G$ is a path, when valuations are 2-additive.
	The problem is also hard when $G$ is a star and valuations are 2-additive.
	Both problems are also hard for dichotomous valuations specified by a  formula of propositional logic.
\end{theorem}
A valuation function $u_i : \calC(V) \to \mathbb R$ is \emph{dichotomous} if $u(X) \in \{0,1\}$ for each bundle $X \in \calC(V)$. (Note that this is different from binary additive valuations, which are dichotomous only if an agent approves at most one item.) A (monotonic) dichotomous valuation function can be specified by a propositional \emph{goal formula} $\varphi$ over the items $V$ using only positive literals, such that $u_i(X) = 1$ if and only if $\varphi$ is satisfied by the variable assignment that sets exactly the variables in $X$ to true. For example, an agent with goal formula $(v_1 \land v_2) \lor v_3$ has positive utility for all bundles $X\in\calC$ with $\{v_1,v_2\}\subseteq X$ or with $v_3 \in X$. For more on such valuations, see \citet{Bouveret2008}.

For a set $X$, let $\calB(X) = \{ Y \subseteq X : 1\le |Y| \le 2 \}$ be the collection of subsets of $X$ (not necessarily connected) of size 1 or 2. A valuation function $u_i : \calC(V) \to \mathbb R$ is \emph{2-additive} if there is a function $w_i : \calB(V) \to \mathbb R$ such that
\[
u_i(X) = \sum_{Y\in \calB(X)} w_i(Y).
\]
2-additive valuation functions allow agents to specify that two items $v_i,v_j$ are complements (if $w_i(\{v_i,v_j\}) \ge 0$) or supplements  (if $w_i(\{v_1,v_2\}) \le 0$). If $w_i(\{v_i,v_j\}) = 0$ for all $v_i,v_j\in V$, then $u_i$ is additive.

\begin{proof}[Proof of Theorem~\ref{thm:PO:hardness:path-star} for a path]
	We give a Turing reduction from X3C similar to Theorem~\ref{thm:forest:Perfect}.
	
	Given an instance $(X, \calS)$ of X3C, with $X=\{x_1,x_2,\dots, x_{3r}\}$ and $\calS = \{S_1,\dots, S_s\}$, we build the following instance.
	For each $S\in\calS$, construct a path $P_S$ on three vertices $v_{S,1}, v_{S,2}, v_{S,3}$.
	We construct the graph $G$ by concatenating the paths $P_{S_1}, \dots, P_{S_s}$ in that order.
	For each $x\in X$, we write $V_x = \{ v_{S,j} \in V : S^j = x \}$. We say that $v,v'\in V_x$ are \emph{consecutive} if $v \neq v'$ and there is no $v'' \in V_x$ that appears in between $v$ and $v'$ on the path $G$.
	
	For each element $x \in X$, we introduce an agent $i_x$ 
	whose most-preferred bundles are those that contain some item $v\in V_x$.
	We also introduce $s-r$ dummy agents $d_1,\dots,d_{s-r}$ whose most-preferred bundles are those which completely contain at least one path $P_S$.
	These preferences can be implemented by 2-additive valuation functions, where
	\begin{align*}
	&w_{i_x}(\{v\}) = 1 &&\text{ for $v\in V_x$}, \\
	&w_{i_x}(\{v, v'\}) = -1 &&\text{ for consecutive $v,v'\in V_x$}, \\
	&w_{d_k}(\{v_{S_j,1}, v_{S_j,3}\}) = 1 &&\text{ for $j \in [s]$,} \\
	&w_{d_k}(\{v_{S_j,1}, v_{S_{j+1},3}\}) = -1 &&\text{ for $j \in [s-1]$}.
	\end{align*}
	Unless we explicitly specify, each agent has zero value for $X \in \calB(V)$.
	We now show that these 2-additive valuation functions correctly implement the above statements about most-preferred bundles.
	For agent $i_x$, if $X\in\calC(V)$ is a bundle containing no item $v\in V_x$, then $u_{i_x}(X) = 0$. On the other hand, if $|X \cap V_x| = q \ge 1$, then $X$ contains $q-1$ pairs of consecutive members of $V_x$ (because $X$ is connected on the path $G$), and hence $u_{i_x}(X) = q - (q-1) = 1$.
	For agent $d_k$, let $X\in\calC(V)$ be a bundle. Note that $P_S \subseteq X$ iff $\{v_{S,1}, v_{S,3}\} \subseteq X$ by connectedness of $X$. Thus, if $X \not\supseteq P_S$ for every $S\in\calS$, then $u_i(X) = 0$. On the other hand, if $X$ contains $k$ paths $P_{S_{j}}, \dots, P_{S_{j + k -1}}$, then $u_i(X) = k - (k-1) = 1$.
	
	Now, suppose we have an algorithm $\calA$ which, given 2-additive valuations, can find a Pareto-optimum on a path. We show that $\calA$ can be used to decide X3C. 
	Run $\calA$ on our allocation instance constructed above to obtain a Pareto-optimal allocation $\pi$. 
	We claim that the X3C instance has a solution if and only if
	\begin{equation}
	\label{eq:path-general-perfect}
	\begin{aligned}
	u_{i_x}(\pi(i_x)) &= 1 \text{ for all $x\in X$ and} \\
	u_{d_k}(\pi(d_k)) &= 1 \text{ for all $k \in [s-r]$}.
	\end{aligned}
	\end{equation}
	
	If \eqref{eq:path-general-perfect} holds, then each bundle $\pi(d_k)$ contains at least one path $P_S$ and hence a total of $s-r$ paths are allocated to agents $d_1,\dots,d_{s-r}$. Hence the items $R = \bigcup_{x\in X} \pi(i_x)$ allocated to the remaining agents are contained within at most $r$ paths corresponding to at most $r$ sets in $\calS$. But by \eqref{eq:path-general-perfect}, $R\cap V_x \neq \emptyset$ for each $x\in X$, hence there is a collection of at most $r$ sets from $\calS$ that cover $X$, and so the X3C instance has a solution.
	
	Suppose $\calS' \subseteq \calS$ is a solution to the X3C instance.
	Define the allocation $\pi^*$ like in Theorem~\ref{thm:forest:Perfect}.
	Then $\pi^*$ satisfies \eqref{eq:path-general-perfect}. 
	Now, if $\pi$ does not satisfy  \eqref{eq:path-general-perfect}, then $\pi^*$ would Pareto-dominate $\pi$, a contradiction. 
	Hence $\pi$ satisfies  \eqref{eq:path-general-perfect}.
	
	The same proof works for dichotomous valuations, since the specification about agents' most-preferred bundles can be specified by propositional logic, where for each $x\in X$, agent $i_x$'s goal formula is $\bigvee_{S,j : S^j = x} v_{S,j}$, and for each $k\in [s-r]$, agent $d_k$'s goal formula is $\bigvee_{s\in [s]} (v_{S_j, 1} \land v_{S_j, 2} \land v_{S_j, 3})$.
\end{proof}

\begin{proof}[Proof of Theorem~\ref{thm:PO:hardness:path-star} for a star]
We give a Turing reduction from {\sc Vertex Cover}, which asks whether given an undirected graph $H=(W,E)$ and positive integer $k$, there is a {\em vertex cover}, i.e., a subset $W'\subseteq W$ of $k$ vertices such that for each edge $\{w_1,w_2\}\in E$, either $w_1\in W'$ or $w_2\in W'$ \citep{gj}.

Given an instance $(H,k)$ of {\sc Vertex Cover}, we take an instance with item set $V = W \cup \{c\}$, an underlying graph $G$ which is a star with center $c$, and the agent set consists of an agent $i$ plus dummy agents $d_1,\dots,d_{|W|-k}$. Agents' preferences are specified so that $i$'s most-preferred bundles are exactly those that contain a vertex cover of $G$, and so that, for each $j\in[|V|-k]$, $d_j$'s most-preferred bundles are exactly the non-empty ones. Such preferences can be implemented by 2-additive valuation functions, where
\begin{align*}
	&w_i(\{w\}) = \text{degree}_H(w) &&\text{ for $w\in W$}, \\
	&w_i(\{w_1,w_2\}) = -1 &&\text{ for $\{w_1,w_2\} \in E$}, \\
	&w_{d_j}(\{o\}) = 1 &&\text{ for $o\in W \cup \{c\}$}, \\
	&w_{d_j}(\{c,w\}) = -1 &&\text{ for $w \in W$}. 
\end{align*}
Unless we explicitly specify, each agent has zero value for $X \in \calB(V)$.
Note that, for a bundle $X\in\calC(V)$, we have $u_i(X) = \sum_{w \in W \cap X} \text{degree}_H(w) - |\{ e \in E : e \subseteq X \}| = |\{ e \in E : e \cap X \neq \emptyset \}|$, that is, $u_i(X)$ is the number of edges covered by the vertices $W\cap X$. Hence, $u_i(X) \le |E|$, with equality iff $X$ contains a vertex cover of $H$. Further, whenever $X\in\calC(V)$ is nonempty, since $G$ is a star, either $X$ is a singleton and then $u_{d_j}(X) = 1$, or $X$ contains $c$ and $u_{d_j}(X) = |X| - |X\setminus \{c\}| = 1$. Hence, these 2-additive valuation functions correctly implement the above statements about most-preferred bundles.

Now, suppose we have an algorithm $\calA$ which, given 2-additive valuations, can find a Pareto-optimum on a star. We show that $\calA$ can be used to decide \textsc{Vertex Cover}. Run $\calA$ on our allocation instance constructed above to obtain a Pareto-optimal allocation $\pi$. We claim that $H$ has a vertex cover of size $k$ if and only if
\begin{equation}
\label{eq:star-general-perfect}
\begin{aligned}
u_{i}(\pi(i)) &= |E| \text{ and} \\
u_{d_j}(\pi(d_j)) &= 1 \text{ for all $j \in [|V|-k]$}.
\end{aligned}
\end{equation}

If \eqref{eq:star-general-perfect} holds, then $W' = \pi(i) \cap W$ is a vertex cover of $H$. We may assume that $|W'| \ge 2$, and hence $c\in\pi(i)$. 
Thus, by connectivity in the star $G$, the bundles $\pi(d_j)$ must be empty or singletons for each $j \in [|V|-k]$.
But by \eqref{eq:star-general-perfect}, the bundles $\pi(d_j)$ are all non-empty. Thus, $|W'| \le k$, and hence there exists a vertex cover of $H$ with size $k$.

Suppose $W'\subseteq W$ is a vertex cover of $H$ with $|W'| = k$.
Define the allocation $\pi^*$ where $\pi^*(i) = W' \cup \{c\}$ and $\pi(d_j)$ is a singleton leaf from $W \setminus W'$ for each $j \in [|V|-k]$.
Then $\pi^*$ satisfies \eqref{eq:star-general-perfect}. 
If $\pi$ does not satisfy  \eqref{eq:star-general-perfect}, then $\pi^*$ would Pareto-dominate $\pi$, a contradiction. 
Hence $\pi$ satisfies  \eqref{eq:star-general-perfect}.

The same proof works for dichotomous valuations, since the specification about agents' most-preferred bundles can be specified by propositional logic, where $i$'s goal formula is $\bigwedge_{\{w_1,w_2\} \in E} (w_1 \lor w_2)$ and $d_j$'s goal formula is $c \lor \bigvee_{w\in W} w$ for each $j\in [|W|-k]$. 
\end{proof}

\section{Pareto-Optimality \& EF1 on Paths}
In Section~\ref{sec:pareto}, we were aiming to find \emph{some} Pareto-optimum, and obtained a positive result for the important case where $G$ is a path. Now we aim higher, wanting to find an efficient allocation which is also fair, where by fairness we mean EF1.

When there are no connectivity requirements, it is known that efficiency and fairness are compatible: \citet{CKM+16a} showed that an allocation maximizing the \textit{Nash product} of agents' valuations is both Pareto-optimal and EF1. While it is NP-hard to compute the Nash solution,  \citet{Barman2018} designed a (pseudo-)polynomial-time algorithm which finds an allocation satisfying these two properties. 

In our model, unfortunately, EF1 is incompatible with Pareto-optimality, even when $G$ is a path. The following examples only require binary additive valuations and at most four agents. 
Note that \citet{Bilo2018} proved that an EF1 allocation always exists on a path for up to four agents.

\begin{example}\label{ex:POEF1}
\upshape
Consider an instance with four agents $a_1,a_2,a_3, b$ and a path with ten items $v_1,\dots,v_{10}$, and binary additive valuations as shown below.
	
	
	
\begin{center}
{\setlength{\tabcolsep}{6.4pt}
\begin{tabular}{rcccccccccc}
	\toprule
	&
	\multicolumn{10}{l}{\!\!\!
		\begin{tikzpicture}[scale=0.51, transform shape, every node/.style={minimum size=8mm, inner sep=1.2pt, font=\huge}]	
		\node[draw, circle](2)  at (1.2,0) {$v_1$};
		\node[draw, circle](3)  at (2.4,0) {$v_2$};
		\node[draw, circle](4)  at (3.6,0) {$v_3$};
		\node[draw, circle](5)  at (4.8,0) {$v_4$};
		\node[draw, circle](6)  at (6,0) {$v_5$};
		\node[draw, circle](7)  at (7.2,0) {$v_6$};
		\node[draw, circle](8) at (8.4,0) {$v_7$};
		\node[draw, circle](9) at (9.6,0) {$v_8$};
		\node[draw, circle](10) at (10.8,0) {$v_9$};
		\node[draw, circle, inner sep=0.4pt](11) at (12.0,0) {$v_{10}$};
		\draw[-, >=latex,thick] (2)--(3) (3)--(4) (4)--(5) (6)--(5) (6)--(7) (7)--(8) (8)--(9) (9)--(10) (10)--(11);	
		\end{tikzpicture}\!\!\!\!
		\vspace{-2pt}
	} \\
	\midrule
	$a_1, a_2,a_3:$\!\! &  1 &  1 & 1 & 1 & 0  & 0 &  1 & 1 &  1& 1 \\
	$b:$\!\! & 0  & 0 & 0 & 0   & 1 &1 & 0 & 0  & 0  & 0  \\
	\bottomrule
\end{tabular}}
\end{center}
\smallskip

\noindent
Suppose $\pi$ is a Pareto-optimal EF1 allocation.
Then, for each $i=1,2,3$, because $b$ does not envy $a_i$ up to one good, we have $\{v_5,v_6\} \not\subseteq \pi(a_i)$. 
Thus, for each $i=1,2,3$, either $\pi(a_i) \subseteq \{v_1,\dots,v_5\}$ (and we say $a_i$ is in group~L) or $\pi(a_i) \subseteq \{v_6,\dots,v_{10}\}$ (and $a_i$ is in group R). Now, $a_1,a_2,a_3$ are not all in group L, since then one of them (say~$a_1$) would receive at most 1 approved item, and there would be a Pareto-improvement by giving the four items $\{ v_7, \dots, v_{10} \}$ to $a_1$. Similarly, $a_1,a_2,a_3$ are not all in group~R. Hence, wlog, two agents (say~$a_1,a_2$) are in group~L and one agent (say~$a_3$) is in group~R.
Since $\pi$ is Pareto-optimal, we have $\pi(b) \subseteq \{v_5, v_6\}$; if $b$ were to obtain any other items (which $b$ does not approve but every other agent does approve), then we can reallocate these items to obtain a Pareto-improvement.
Thus, $a_3$ obtains four approved items, but one of $a_1$ or $a_2$ obtains at most two approved items, so $\pi$ is not EF1, a contradiction.
\qed
\end{example}

The following alternative example shows that Pareto-optimality and EF1 conflict in an even more restricted setting, where each agent's approval set is an interval.

\begin{example}\label{ex2:POEF1}
\upshape
Consider an instance with three agents $a_1,a_2$, and $b$, and a path with eleven items $v_1,\dots,v_{11}$, and binary additive valuations as shown below.

\begin{center}
{ \setlength{\tabcolsep}{6.4pt}
	\begin{tabular}{rccccccccccc}
		\toprule
		&
		\multicolumn{11}{l}{\!\!\!
			\begin{tikzpicture}[scale=0.51, transform shape, every node/.style={minimum size=8mm, inner sep=1.2pt, font=\huge}]	
			\node[draw, circle](2)  at (1.2,0) {$v_1$};
			\node[draw, circle](3)  at (2.4,0) {$v_2$};
			\node[draw, circle](4)  at (3.6,0) {$v_3$};
			\node[draw, circle](5)  at (4.8,0) {$v_4$};
			\node[draw, circle](6)  at (6,0) {$v_5$};
			\node[draw, circle](7)  at (7.2,0) {$v_6$};
			\node[draw, circle](8) at (8.4,0) {$v_7$};
			\node[draw, circle](9) at (9.6,0) {$v_8$};
			\node[draw, circle](10) at (10.8,0) {$v_9$};
			\node[draw, circle, inner sep=0.4pt](11) at (12.0,0) {$v_{10}$};
			\node[draw, circle, inner sep=0.4pt](12) at (13.2,0) {$v_{11}$};
			\draw[-, >=latex,thick] (2)--(3) (3)--(4) (4)--(5) (6)--(5) (6)--(7) (7)--(8) (8)--(9) (9)--(10) (10)--(11) (11)--(12);	
			\end{tikzpicture}\!\!\!\!
			\vspace{-2pt}
		} \\
		\midrule
		$a_1, a_2:$\!\! &  1 &  1 & 1 & 1 & 1 & 1 &  1 & 1 & 1 & 1 & 1 \\
		$b:$\!\! & 0  & 0 & 0 & 1 &1 & 0 & 0 & 0 & 0 & 0 & 0  \\
		\bottomrule
\end{tabular}}
\end{center}
\smallskip

\noindent
Suppose $\pi$ is a Pareto-optimal EF1 allocation.
Then, for each $i=1,2$, because $b$ does not envy $a_i$ up to one good, we have $\{v_4,v_5\} \not\subseteq \pi(a_i)$. 
Thus, for each $i=1,2$, we have either $\pi(a_i) \subseteq \{v_1,\dots,v_4\}$ (and $a_i$ is in group L) or $\pi(a_i) \subseteq \{v_5,\dots,v_{11}\}$ (and $a_i$ is in group R). 
Now, $a_1$ and $a_2$ are not both in group L, since then there would be a Pareto-improvement by giving the six items $\{ v_6, \dots, v_{11} \}$ to $a_1$. Also, $a_1$ and $a_2$ are not both in group R, since then one of them (say $a_1$) would receive at most 3 approved items, and there would be a Pareto-improvement by giving items $\{ v_1, v_2, v_3 \}$ to $a_1$ and $\{ v_6, \dots, v_{11} \}$ to $a_2$.
Hence, wlog, $a_1$ is in group~L and $a_2$ is in group~R.
Since $\pi$ is Pareto-optimal, we have $\pi(b) \subseteq \{v_4, v_5\}$; if $b$ were to obtain any other items, then we can reallocate these items to $a_1$ and $a_2$ to obtain a Pareto-improvement.
Thus, $a_1$ obtains at most four approved items (since $\pi(a_1) \subseteq \{v_1,\dots,v_4\}$), but $a_2$ receives at least six approved items (since $\{v_6,\dots,v_{11}\} \subseteq \pi(a_2)$), so $\pi$ is not EF1, a contradiction.
\qed
\end{example}

Given that we do not have an existence guarantee, a natural question is whether it is easy to decide whether a given instance admits a Pareto-optimal allocation satisfying EF1. Using the above examples, we prove that the problem is NP-hard. The obvious complexity upper bound is $\Sigma_2^p$; an open problem is whether the problem is complete for this class. A related result of \citet{deKeijzer2009} shows that without connectivity constraints and with additive valuations, it is $\Sigma_2^p$-complete to decide whether a Pareto optimal and envy-free allocation exists; see also \citet{Bouveret2008}.

\begin{theorem}\label{thm:POEF1:hardness:path}
	It is \textup{NP}-hard to decide whether a Pareto-optimal EF1 connected allocation exists when $G$ is a path, even if valuations are binary and additive.
\end{theorem}
\begin{proof}
Given an instance $(X, \calS)$ of X3C where $X=\{x_1,x_2,\ldots,x_{3r}\}$ and $\calS=\{S_1,S_2,\ldots,S_s\}$, we again create the same instance $I$ as in the proof of Theorem \ref{thm:forest:Perfect}. Namely, we make a path $P_S =({v}_{S,1}, v_{S,2}, {v}_{S,3})$ for each $S  \in  \calS$, and construct agent $i_x$ for each $x \in X$ and agents ${d}_1,d_2,\dots,{d}_{s-r}$ with the same binary valuations:
\begin{itemize}
\item $u_{i_x}(v) = 1$ iff $v=v_{S,j}$ and $S^j = x$;
\item $u_{{d}_k}(v) = 1$ iff $v=v_{S,j}$ for some $S,j$. 
\end{itemize}
We write $N_o=\{\, i_x \mid x \in X\,\} \cup \{d_1,d_2,\ldots,d_{s-r}\}$ for the set of the original agents, and $V_o=\bigcup_{S \in \calS}\{v_{S,1},v_{S,2},v_{S,3}\}$ for the set of original vertices. We will create additional items and agents as follows; unless we explicitly specify, each original agent has zero value for the additional items.

{\em $I_{d_k}$ gadget}: We create an empty PO and EF1 instance $I_{d_k}$ for each agent $d_k$ $(k \in [s-r])$. The empty instance $I_{d_k}$ consists of a path $P_{d_k}=(v^k_1,v^k_2,\ldots,v^k_{10})$ of ten vertices as depicted in Example \ref{ex:POEF1}, together with agent ${d_k}$, agents $a^1_k$ and $a^2_k$, and agent $b_k$. Each of $d_k$, $a^1_k$, and $a^2_k$ approves the vertices of the path except for $v^k_5$ and $v^k_6$ while $b_k$ approves $v^k_5$ and $v^k_6$ only. The agents $a^1_k$, $a^2_k$, and $b_k$ do not approve the vertices outside of $P_{d_k}$. See below for these valuations. 

\begin{center}
{
	\setlength{\tabcolsep}{6.6pt}
	\begin{tabular}{rcccccccccc}
		\toprule
		\raisebox{5pt}{$I_{d_k}:$}\!\!&
		\multicolumn{10}{l}{\!\!\!
			\begin{tikzpicture}[scale=0.51, transform shape, every node/.style={minimum size=8mm, inner sep=1.2pt, font=\huge}]	
			\node[draw, circle](2)  at (1.2,0) {$v_1^k$};
			\node[draw, circle](3)  at (2.4,0) {$v_2^k$};
			\node[draw, circle](4)  at (3.6,0) {$v_3^k$};
			\node[draw, circle](5)  at (4.8,0) {$v_4^k$};
			\node[draw, circle](6)  at (6,0) {$v_5^k$};
			\node[draw, circle](7)  at (7.2,0) {$v_6^k$};
			\node[draw, circle](8) at (8.4,0) {$v_7^k$};
			\node[draw, circle](9) at (9.6,0) {$v_8^k$};
			\node[draw, circle](10) at (10.8,0) {$v_9^k$};
			\node[draw, circle, inner sep=0.3pt](11) at (12.0,0) {$v_{10}^k$};
			\draw[-, >=latex,thick] (2)--(3) (3)--(4) (4)--(5) (6)--(5) (6)--(7) (7)--(8) (8)--(9) (9)--(10) (10)--(11);	
			\end{tikzpicture}\!\!\!\!
			\vspace{-2pt}
		} \\
		\midrule
		$d_k, a^1_k, a^2_k:$\!\! & 1 & 1 & 1 & 1 & 0 & 0 & 1 & 1 & 1 & 1 \\
		$b_k:$\!\! & 0 & 0 & 0 & 0 &1 & 1 & 0 & 0 & 0 & 0  \\
		\bottomrule
\end{tabular}}
\end{center}
\smallskip

{\em $I_x$ gadget}: We create an empty PO and EF1 instance $I_x$ for each agent $i_x$ $(x \in X)$. The empty instance $I_x$ consists of a path $P_x=(v^x_1,v^x_2,\ldots,v^x_{11})$ as in Example \ref{ex2:POEF1}, together with agent $i_x$ and agents $a^1_x$ and $a^2_x$. Each $i_x$ approves the vertices $v^x_4$ and $v^x_5$ whereas agents $a^1_x$ and $a^2_x$ approve every vertex on the path but does not approve any other item outside $P_x$. See below for these valuations.

\begin{center}
{
	\setlength{\tabcolsep}{6.5pt}
	\begin{tabular}{rccccccccccc}
		\toprule
		\raisebox{5pt}{$I_x:$}\!\!&
		\multicolumn{11}{l}{\!\!\!
			\begin{tikzpicture}[scale=0.51, transform shape, every node/.style={minimum size=8mm, inner sep=1.2pt, font=\huge}]	
			\node[draw, circle](2)  at (1.2,0) {$v_1^x$};
			\node[draw, circle](3)  at (2.4,0) {$v_2^x$};
			\node[draw, circle](4)  at (3.6,0) {$v_3^x$};
			\node[draw, circle](5)  at (4.8,0) {$v_4^x$};
			\node[draw, circle](6)  at (6,0) {$v_5^x$};
			\node[draw, circle](7)  at (7.2,0) {$v_6^x$};
			\node[draw, circle](8) at (8.4,0) {$v_7^x$};
			\node[draw, circle](9) at (9.6,0) {$v_8^x$};
			\node[draw, circle](10) at (10.8,0) {$v_9^x$};
			\node[draw, circle, inner sep=0.3pt](11) at (12.0,0) {$v_{10}^x$};
			\node[draw, circle, inner sep=0.3pt](12) at (13.2,0) {$v_{11}^x$};
			\draw[-, >=latex,thick] (2)--(3) (3)--(4) (4)--(5) (6)--(5) (6)--(7) (7)--(8) (8)--(9) (9)--(10) (10)--(11) (11)--(12);	
			\end{tikzpicture}\!\!\!\!
			\vspace{-2pt}
		} \\
		\midrule
		$a^1_x, a^2_x:$\!\! &  1 &  1 & 1 & 1 & 1 & 1 &  1 & 1 & 1 & 1 & 1 \\
		$i_x:$\!\! & 0  & 0 & 0 & 1 &1 & 0 & 0 & 0 & 0 & 0 & 0  \\
		\bottomrule
\end{tabular}}
\end{center}
\smallskip

{\em Dummies}: Note that the number of connected components of the graph constructed so far is $s+3r+(s-r)=2s+2r$. For each $h \in [2s+2r]$, we create a dummy agent $z_h$ and a path $P_{z_h}=(v^{z_h}_1,v^{z_h}_2)$ of two dummy vertices. 

Define the following three paths by concatenating pieces constructed above:
\begin{itemize}
\item $P_1 :=P_{S_1}P_{z_1}\cdots P_{z_{s-1}}P_{S_s}P_{z_{s}}$,
\item $P_2 := P_{x_1}P_{z_{s+1}}\cdots P_{z_{s+3r-1}}P_{x_{3r}}P_{z_{s+3r}}$,
\item $P_3 := P_{d_1}P_{z_{s+3r+1}}\cdots P_{z_{2s+2r-1}}P_{d_{s-r}}P_{z_{2s+2r}}$.
\end{itemize}
Our final graph $G$ is obtained by concatenating the paths $P_1$, $P_2$, $P_3$ in that order.

Each dummy agent $z_h$ only approves the dummy vertices on the path $P_{z_h}$. Thus in any EF1 allocation, none of the non-dummy agent obtains a bundle containing both of the dummy vertices. 

{\em Correctness}: We will show that there is a Pareto-optimal EF1 connected allocation if and only if there is an exact cover. First, suppose that there is a Pareto-optimal EF1 connected allocation $\pi$. By EF1, each non-dummy agent cannot obtain more than one connected component of the original graph. 
Thus, we may assume that for each agent $i \neq z_h$, the bundle $\pi(i)$ is contained in some of the non-dummy paths in which $i$ has approved vertices, i.e., 
\begin{itemize}
\item for each $x \in X$, $\pi(i_x) \subseteq P_{\alpha}$ for some $\alpha \in \calS(x) \cup \{x\}$ where $\calS(x)$ is the set of $S \in \calS$ containing $x$;
\item for each $k \in [s-r]$, $\pi(d_k) \subseteq P_{\alpha}$ for some $\alpha \in \calS \cup \{d_k\}$;
\item for each $x \in X$ and $j=1,2$, $\pi(a^j_x) \subseteq P_{x}$;
\item for each $k \in [s-r]$ and $j=1,2$, $\pi(a^j_k) \subseteq P_{d_k}$;
\item for each $k \in [s-r]$, $\pi(b_k) \subseteq P_{d_k}$.
\end{itemize}

Now observe that for each $x \in X$, none of the agents except for $i_x$ approve the vertices outside $P_x$. As we saw in Example \ref{ex2:POEF1}, if $\pi(i_x) \subseteq P_x$, and $\pi(a^j_x) \subseteq P_{x}$ for each $j=1,2$, then $\pi$ would not satisfy both Pareto-optimality and EF1. Hence, in the Pareto-optimal and EF1 allocation $\pi$, each agent $i_x$ receives a bundle outside of $P_x$, namely,  $\pi(i_x) \subseteq P_{\alpha}$ for some $\alpha \in \calS(x)$. Also, to achieve EF1 and PO allocation of the path $P_x$ among the agents $a^1_x$ and $a^2_x$, we must allocate the five consecutive vertices of $P_x$ to one of the agents $a^1_x$ and $a^2_x$, and allocate the remaining six vertices of the path to the other. By EF1, it follows that each agent $i_x$ receives a bundle containing at least one item which he approves.

Similarly, for each $k \in [s-r]$, none of the agents $a^1_k$, $a^2_k$, and $b_k$ approves the vertices outside of $P_{d_k}$. Thus, we must have $\pi(d_k) \subseteq P_{\alpha}$ for some $\alpha \in \calS$; otherwise, $\pi$ would not be Pareto-optimal and EF1 as we have seen in Example \ref{ex:POEF1}. Also, by EF1, no agent other than $b_k$ can get both of the middle two vertices $v^k_5$ and $v^k_6$ which $b_k$ approves; thus, by Pareto-optimality, each agent $b_k$ is allocated the middle two vertices $v^k_5$ and $v^k_6$ of $P_{d_k}$. Thus, by Pareto-optimality, one of the agents $a^1_k$ and $a^2_k$ is allocated the first four vertices of $P_{d_k}$ and the other is allocated the last four vertices of $P_{d_k}$. It follows that to bound the envy of $d_k$ up to one item, each agent $d_k$ receives a bundle containing at least three items which he approves. 

Combining the above observations, each of the original agents is allocated to some of the original vertices and the allocation $\pi$ satisfies 
	\begin{equation}
	\label{eq:POEF1}
	\begin{aligned}
	u_{i_x}(\pi(i_x)) &\ge 1 \text{ for all $x\in X$ and } \\
	u_{d_k}(\pi(d_k)) &\ge 3 \text{ for all $k \in [s-r]$}.
	\end{aligned}
	\end{equation}
Hence there is an exact cover as we have seen in the proof of Theorem \ref{thm:forest:Perfect}. 

Conversely, suppose that there is an exact cover. Then, as we proved in the proof of Theorem \ref{thm:forest:Perfect}, there is a perfect allocation $\pi$ of the original instance satisfying the inequalities \ref{eq:POEF1}. We extend this allocation as follows 
\begin{itemize}
\item each dummy agent $z_h$ obtains the associated dummy vertices in $P_{z_h}$; 
\item for each $k \in [s-r]$, $a^1_k$ receives the first four vertices of the path $P_{d_k}$, $a^2_k$ receives the last four vertices, and $b_k$ obtains the two vertices $v^k_5$ and $v^k_6$; and
\item for each $x \in X$, $a^1_x$ the first five vertices of the path $P_x$, and $a^2_x$ obtains the last six vertices. 
\end{itemize}
The resulting allocation is Pareto-optimal since each item is allocated to an agent who approves it; also, it can be easily seen to satisfy EF1. 
\end{proof}

Observe that in the Examples~\ref{ex:POEF1} and~\ref{ex2:POEF1}, there are at least two different types of agents' valuations. One may expect that there are existence guarantees when agents have {\em identical valuations}, i.e., $u_i(X)=u_i(X)$ for all bundles $X \in \calC(V)$ and all $i,j \in N$. Invoking a very recent result independently obtained by \citet{Bilo2018} and \citet{Oh2018}, we can show that this is the case for additive valuations: an EF1 and Pareto-optimal allocation exists on paths for agents with identical additive valuations.

\begin{proposition}
When $G$ is a path and agents have identical additive valuations, a connected allocation that satisfies EF1 and Pareto-optimality exists and can be found efficiently. 
\end{proposition}
\begin{proof}
When agents have identical additive valuations, every allocation $\pi$ has the same utilitarian social welfare $\sum_{i \in N} u_i(\pi(i)) = \sum_{v\in V} u_1(v)$. Hence, every allocation maximises social welfare and is thus Pareto-optimal. Now, \citet[Theorem~7.1]{Bilo2018} and \citet[Lemma~C.2]{Oh2018} show that if $G$ is a path, a connected EF1 allocation exists, which, by the above reasoning, is also Pareto-optimal. 
This allocation can be found efficiently since the existence results of \citet{Bilo2018} and \citet{Oh2018} both come with an efficient algorithm for finding an EF1 allocation.
\end{proof}

\noindent
For identical valuations that are not additive, Pareto-optimality and EF1 are again incompatible on a path. The following example uses two agents and four items, and subadditive valuations.
\begin{example}[PO and EF1 may be incompatible for identical but non-additive valuations]
We give an example with identical non-additive valuations, for which no Pareto-optimal allocation is EF1. There are four items $a,b,c,d$ arranged on a path, and two agents with the following valuations:
\begin{center}
\begin{tabular}{rccrc}
	\toprule
	$X$ & $u(X)$ &\:& $X$ & $u(X)$ \\
	\midrule
	$\emptyset$ & $0$ && $\{a,b\}$ & $2$ \\
	$\{a\}$ & $2$ && $\{b,c\}$ & $3$ \\
	$\{b\}$ & $2$ && $\{c,d\}$ & $3$ \\
	$\{c\}$ & $2$ && $\{a,b,c\}$ & $3$ \\
	$\{d\}$ & $1$ && $\{b,c,d\}$ & $4$ \\
	& && $\{a,b,c,d\}$ & $4$ \\
	\bottomrule
\end{tabular}
\end{center}
Then:
\begin{itemize}
	\item allocation $\{ \{a,b,c,d\} \}$ is not EF1;
	\item allocation $\{ \{a,b,c\}, \{d \} \}$ is not EF1;
	\item allocation $\{ \{a,b\}, \{c, d \} \}$ is Pareto-dominated by $\{ \{a\}, \{b,c,d \} \}$;
	\item allocation $\{ \{a\}, \{b,c,d \} \}$ is not EF1.
\end{itemize}
One can check that these valuations are subadditive.
\qed
\end{example}

\section{Pareto-Optimality \& MMS on Paths}
In the previous section, we saw that deciding the existence of an allocation that is Pareto-efficient and satisfies EF1 is computationally hard, even if $G$ is a path, and there are examples where no such allocation exists. Part of the reason is that envy-freeness notions and Pareto-optimality are not natural companions: it is easy to construct examples where some allocation is envy-free, yet by Pareto-improving the allocation, we introduce envy.

An alternative notion of fairness avoids this problem: Pareto-improving upon an MMS allocation preserves the MMS property, because MMS only specifies a lower bound on agents' utilities.
\citet{Bouveret2017} showed that if $G$ is a tree, then an MMS allocation always exists (and can be found efficiently). Hence, if $G$ is a tree, there is an allocation that is both Pareto-optimal and MMS: take an MMS allocation, and repeatedly find Pareto-improvements until reaching a Pareto-optimum, which must still satisfy the MMS property.

While existence is guaranteed, it is unclear whether we can find an allocation satisfying both properties in polynomial time. Certainly, by the negative result of Theorem~\ref{thm:PO:hardness:tree}, this is not possible when $G$ is an arbitrary tree. What about the case when $G$ is a path? The answer is also negative: a Pareto-optimal MMS allocation cannot be found efficiently.

\begin{theorem}\label{thm:POMMS:hardness:path}
	Unless \textup{P = NP}, there is no polynomial-time algorithm which finds a Pareto-optimal MMS allocation when $G$ is a path, even if valuations are binary and additive.
\end{theorem}
\begin{proof}
	We again give a Turing reduction from X3C, building on the reduction of Theorem~\ref{thm:forest:Perfect}.
	Suppose we are given an instance $(X,\calS)$ of X3C, where $X=\{x_1,x_2,\dots, x_{3r}\}$ and $\calS = \{S_1,\dots, S_s\}$. Construct the paths $P_{S_1},P_{S_2},\dots, P_{S_s}$ and agents $i_x$ for each $x\in X$ and $d_k$ for each $k \in [s-r]$ with binary utilities like in the proof of Theorem~\ref{thm:forest:Perfect}.
	We write $N_{o}=\{\, i_x :  x \in X\,\} \cup \{d_1,d_2,\ldots,d_{s-r}\}$ and $V_o = \bigcup_{S \in \calS} \{ v_{S,1}, v_{S,2}, v_{S,3} \}$ for the sets of agents and items introduced so far.
	
	In addition, for each $k \in [s]$, we construct a path $B_k$ of $2r+2s$ new vertices $b_k^1,b_k^2, \dots, b_k^{2r+2s}$. 
	The graph $G$ for our problem is obtained by concatenating these paths in the order $P_{1},B_{1},\dots,P_{s},B_{s}$.
	Finally, for each $k \in [s]$, we introduce an agent $z_k$ who approves exactly the vertices on $B_k$.
	The agents in $N_o$ do not approve any of the items in $B_1,\dots,B_s$.
	
	Note that, in total, there are $3r + (s-r) + s = 2r+2s$ agents. Since each agent $z_k$ approves $2r+2s$ vertices, each agent $z_k$ has positive MMS value, namely $\mms_{z_k} = 1$. 
	
	Suppose we had an algorithm $\calA$ which finds a Pareto-optimal MMS allocation on a path. We show how to use $\calA$ to solve X3C.
	Run $\calA$ on the allocation problem constructed above to obtain a Pareto-optimum $\pi$ which satisfies MMS. 
	Then, for each $k\in [s]$, the agent $z_k$ receives at least one vertex from $B_k$ since $\pi$ is MMS.
	It follows that no agent $i\in N_0$ can receive items from two different paths $P_{S_j}$ and $P_{S_k}$, $j<k$, since these paths are separated by $B_{j}$.
	Thus, for each $i \in N_0$, there is some $j\in [s]$ with $\pi(i) \subseteq B_{j-1} \cup P_{S_j} \cup B_{j}$.
	By suitably reallocating items that agent $i$ does not approve, we can in fact assume that $\pi(i) \subseteq P_{S_j}$ for some $j\in [s]$.
	This implies that $u_{i_x}(\pi(i_x)) \le 1$ for all $x\in X$ and $u_{d_k}(\pi(d_k)) \le 3$ for all $k \in [s-r]$.
	
	We now prove that the X3C instance has a solution iff
	\begin{equation}
	\label{eq:mms-perfect}
	\begin{aligned}
	u_{i_x}(\pi(i_x)) &= 1 \text{ for all $x\in X$ and } \\
	u_{d_k}(\pi(d_k)) &= 3 \text{ for all $k \in [s-r]$}.
	\end{aligned}
	\end{equation}
	Since \eqref{eq:mms-perfect} is easy to check, this equivalence implies that $\calA$ can be used to solve X3C, and hence our problem is NP-hard.
	
	If \eqref{eq:mms-perfect} holds, then the argument in the proof of Theorem~\ref{thm:forest:Perfect} applies and shows that the X3C instance has a solution.
	
	Conversely, suppose there is a solution $\calS' \subseteq \calS$ to the X3C instance.
	Then, as in the proof of Theorem \ref{thm:forest:Perfect}, there is an allocation $\pi^* : N_o \to \calC(V_o)$ of the original items to the original agents such that 
	$u_{i_x}(\pi^*(i_x)) = 1$ for all $x\in X$ and
	$u_{d_k}(\pi^*(d_k)) = 3$ for all $k \in [s-r]$.
	Extend $\pi^*$ to all agents by defining $\pi^*(z_k) = B_k$ for each $k\in[s]$. 
	It is easy to check that $\pi^*$ is a connected allocation. 
	For each $k\in[s]$, we have $u_{z_k}(z_k) = u_{z_k}(\pi(z_k))$, since $z_k$ receives all approved items in $\pi^*$
	Also, for each original agent $i\in N_o$, we have $u_i(\pi^*(i)) \ge u_i(\pi(i))$, since $i$ obtains an optimal bundle under $\pi^*$.
	It follows that if $\pi$ fails \eqref{eq:mms-perfect}, then $\pi^*$ is a Pareto-improvement of $\pi$, contradicting that $\pi$ is Pareto-optimal. So $\pi$ satisfies \eqref{eq:mms-perfect}.
\end{proof}

For $\alpha \in (0,1]$, we say that an allocation $\pi$ is $\alpha$-MMS if $u_i(\pi(i)) \ge \alpha\cdot \mms_i$ for all $i\in N$. 
The above proof implies that we cannot in polynomial time find a Pareto-optimal allocation that is $\alpha$-MMS, for any $\alpha > 0$.
The reduction can also easily be adapted to the case when $G$ is a cycle.

Next, we show that when $G$ is a path, we can find a Pareto-optimal MMS allocation in polynomial time for a restricted class of valuations. We assume that agents' valuations are binary and additive, and for each voter, the set of approved vertices forms an interval of the path $G$, and finally these intervals are {\em non-nested}. Formally, given binary and additive valuations and an agent $i \in N$, we let $A(i) = \{v \in V : u_i(v) = 1\}$ be the set of vertices which $i$ approves. For a path $P=(1,2,\ldots,m)$, we say that binary valuations are {\em non-nested} on the path if for each $i \in N$, $A(i)$ is connected on the path, and there is no pair of agents $i,j \in N$ with $\min A(i) < \min A(j)$ and $\max A(j)< \max A(i)$. The corresponding restriction captures, for instance, when several groups wish to book the same conference venue; each group specifies a period of contiguous dates of (almost) equal length that are suitable for them. We show that when valuations have this form, there is a polynomial-time algorithm which yields a Pareto-optimal MMS allocation. The algorithm is an adaptation of the moving-knife algorithm of \citet{Bouveret2017}.

We first observe that if agent $i$'s approval interval appears before $j$' approval interval and these intervals intersect with each other, then the value of the left bundle exceeds $i$'s maximin fair share before it exceeds $j$'s maximin fair share. For $j,k\in [m]$ with $j\leq k$, we write $[j,k]=\{j,{j+1},\ldots,{k}\}$.  

\begin{lemma}\label{lem:interval}
Suppose that $G$ is a path $P=(1,2,\ldots,m)$ and valuations are binary and additive. Let $i,j \in N$ be a pair of distinct agents such that $A(i)$ and $A(j)$ are connected on the path, $\min A(i) \leq \min A(j)$, and $\max A(i) \leq \max A(j)$. If agent $j$ values $[1,x]$ at least as highly as her maximin fair share and $\min A(j) \leq x+1$, then $i$ values $[1,x]$ at least as highly as her maximin fair share.  
\end{lemma}
\begin{proof}
If the interval $[1,x]$ contains $i$'s approval interval, i.e., $\max A(i) \le x$, then $i$'s value for $[1,x]$ is at least her maximin fair share. Thus suppose that the last vertex of $i$'s approval interval appears after or at $x+1$ (i.e., $x+1 \leq \max A(i)$). Assume for a contradiction that agent $i$ does not value $[1,x]$ at least as highly as her maximin fair share, i.e., $u_i([1,x]) < \mms_i$. Since the $\min A(j)  \leq x+1$, we have $\min A(i) \leq \min A(j) \leq x+1$; also, since $x+1 \leq \max A(i)$, we have that $x+1 \leq \max A(i) \leq \max A(j)$. By the contiguity, it follows that both $i$ and $j$ approves the item $x+1$. Now let $s$ denote the number of vertices approved by $i$ in $[1,x]$. Similarly, let $t$ denote the number of vertices approved by $j$ in $[1,x]$. By the above argument, note that $s \ge t$.
Since $[1,x]$ does not guarantee $i$'s maximin fair share, $i$ has value more than $(n-1)(s+1)$ for $[x+1,m]$, i.e., 
$$|[x,m] \cap A (i)|> (n-1)(s+1).$$
Otherwise, $i$'s maximin fair share would be at most $s$. However, since $j$'s value for $[1,x]$ is at least his maximin fair share, $j$ has value at most $(n-1)(t+1)$ for $[x+1,m]$, i.e., 
$$|[x,m]  \cap A (j)| \le (n-1)(t+1).$$
Otherwise, $j$'s maximin fair share would be more than $t$. Now recall that $s \ge t$, which implies $(n-1)(s+1) \ge (n-1)(t+1)$. This means that agent $i$ has strictly more approved vertices in $[x+1,m]$ than agent $j$, namely, $x+1 \leq \max A(j) < \max A(i)$, a contradiction.
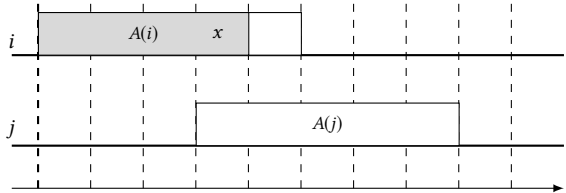
\begin{figure*}[htb]
\centering
\begin{tikzpicture}[scale=0.7,transform shape]
\draw [-latex](-0.5,0) coordinate(dd)-- (0,0) coordinate (O1) -- (10,0)coordinate(ff) node[above]{};
\draw [dashed,thick] (O1) -- (0,0.8) coordinate(B) -- (0,2.5) coordinate(A) -- ++(0,1)coordinate(ff2);

\draw [thick] (dd|-A) node[above]{\Large $i$}-- (A-|ff);
\draw [thick] (dd|-B) node[above]{\Large $j$}-- (B-|ff);

\foreach \xx in{1,2,...,9}{
\draw[dashed] (\xx,0) -- (\xx,0|- ff2);
}

\begin{scope}[shift={(A)}]
\node[above right=0.4cm and 0cm of A,right,draw, minimum width=5cm,minimum height=0.8cm,fill=white](n3a) {};
\node[above right=0.4cm and 0cm of A,right,draw, minimum width=4cm,minimum height=0.8cm,fill=gray!30](n3b) {$A(i)$};
\node[above right=0.19cm and 3.2cm of A](n3c) {\large $x$};
\end{scope}

\begin{scope}[shift={(B)}]
\coordinate(O2) at (3,0);
\node[above right=0.4cm and 0cm of O2,right,draw, minimum width=5cm,minimum height=0.8cm,fill=white](n4a) {$A(j)$};
\end{scope}
\end{tikzpicture}
\caption{Illustration of non-nested approval intervals. }\label{fig:nonnested}
\end{figure*}
\end{proof}

We are now ready to show that a moving-knife algorithm can be used to produce a Pareto-optimal MMS allocation. Intuitively, that following algorithm sequentially creates a connected bundle for $i_1$, $i_2$, and so on, such that we add one vertex to the bundle of $i_j$ as long as this bundle does not exceed the maximin fair share of the agent $i_{j+1}$. By doing so, we can ensure that each agent receives a connected bundle of at least maximin fair share and every item has been allocated to an agent who approves of it if there is any. Hence, the resulting allocation maximizes the utilitarian social welfare and so is Pareto-optimal. 

In what follows, we denote by $\mms_i(I')$ the maximin fair share of $i$ for the instance $I'$ with subpath $P'$, agent set $N'$, and valuations $(u_i)_{i \in N'}$. Namely, 
\[
\mms_i(I') = \max_{(A_1, \dots, A_{n'})\in\Pi_{n'}(P')}\min_{j\in [n']} u_i(A_j). 
\]
for each $i \in N'$ where $n'=|N'|$.

\begin{theorem}\label{thm:PO:MMS:approval}
When the graph $G$ is a path and valuations are binary and additive given by non-nested intervals on the path, there exists a polynomial-time algorithm that finds an MMS connected allocation that maximizes the utilitarian social welfare.  
\end{theorem}
\begin{proof}
We design a recursive algorithm $\calA$ that takes as input the instance $I'$ with a subset $N'$ of agents, a subpath $P'=(1,2,\ldots,m)$, and a valuation profile $(u_i)_{i \in N'}$, and returns a utilitarian optimal MMS allocation of the items to the agents in $N'$. We write $n':=|N'|$. Assume without loss of generality that that every item is approved by some agent, and every agent approves some item. 
If $n'=1$, then the algorithm allocates all the items to the single agent. Suppose that $n'>1$. 
\begin{enumerate}
	\item Order agents as $i_1,i_2,\dots,i_{n'}$ so that $\min A(i_k) \le \min A(i_{k+1})$ and if $\min A(i_k) = \min A(i_{k+1})$ then $\max A(i_k) \le \max A(i_{k+1})$. Hence $A({i_1})$ \emph{ends earliest}, i.e. $\max A({i_1}) \le \max A({i_k})$ for all $k>1$.
	\item Set $x$ to the minimum such that $u_{i_1}([1,x]) \geq \mms_{i_1}(I')$. 
	\begin{itemize}
		\item If $[1,x] \cap A(i_2) \neq \emptyset$, allocate $[1,x]$ to agent $i_1$.
		\item If $[1,x] \cap A(i_2) = \emptyset$, allocate $[1, \min A(i_2) - 1]$ to agent $i_1$.
	\end{itemize}
	\item Recurse, by relabeling items so that the left-most item is again called 1. 
\end{enumerate}
See Figures \ref{fig:nonnested:alg} for an illustration. Note that the reduced instance still remains non-nested since otherwise there would be a pair of agents $i,j$ in the original instance with $\min A(i)< \min A(j)$ and $\max A(j)<\max A(i)$.  

One can compute the maximin fair share of each agent in polynomial time if the graph is a path \citep{Bouveret2017}; thus it is immediate that $\calA$ runs in polynomial time. Now we will prove by induction on the number of agents the following:
\begin{itemize}
\item agent $i \in N'$ receives a connected bundle of value at least $\mms_i(I')$; and
\item every {\em approved} item $v$ (i.e., some $i \in N'$ approves $v$) has been allocated to some agent in $N'$ approving $v$.
\end{itemize}

The claim is immediate for $|N'|=1$. Suppose that the claim holds for $|N'|=\ell-1$; we will prove it for $|N'|=\ell$. Let $\pi$ denote the allocation returned by the algorithm. To show that every approved item has been allocated to some player, observe first that agent $i_1$ has positive value for every approved item in $\pi(i_1)$. Indeed, by construction of the algorithm, agent $i_1$ does not receive any item $y$ that appears after $\max A(i_1)$; and no agent $i_k$ with $k>1$ approves item $y$ that appears before $\min A(i_1)$. Also, there is no item $y \not \in \pi(i_1)$ only approved by agent $i_1$: If there is such an item $y \not \in \pi(i_1)$, then it means that $x<y$ but $y \in A(i_1) \setminus A(i_{2})$, and thereby 
\[
x<y< \min A(i_{2}) \leq \max A(i_1).
\]
Hence, $[1,x] \cap A(i_2) = \emptyset$ and $y \in \pi(i_1)$, a contradiction. Applying the induction hypothesis, every approved item not in $\pi(i_1)$ is allocated to some player in $N' \setminus \{i_1\}$ who approves that item. 

Now it remains to prove that every agent receives a bundle of value at least his maximin fair share. Clearly, agent $i_1$ receives a bundle of value at least $\mms_{i_1}(I')$. Take any $i_k$ with $k>1$. Let $I''$ denote the reduced instance after the bundle for $i_1$ is removed. It is clear that $\mms_{i_k}(I'') \ge \mms_{i_k}(I')$ if $i_k$ approves no vertex in $\pi(i_1)$ or $\mms_{i_k}(I')=0$. 
So suppose otherwise, i.e., agent $i_k$ approves at least one vertex in $\pi(i_1)$ and has positive maximin fair share $\mms_{i_k}(I')>0$. 
Now let $X_1,X_2,\ldots,X_{n'}$ be a partition of $P'$ into connected bundles witnessing $\mms_{i_k}(I')$, i.e., $\mms_{i_k}(I')=\min_{h \in [n']} u_{i_k} (X_h)$. Without loss of generality, we assume that the left most item $1$ belongs to $X_1$. 
Since $u_{i_k}(X_1)\geq \mms_{i_k}(I')>0$, $i_k$ approves at least one vertex in $X_1$, thereby implying that the initial vertex of $i_k$'s approval interval appears before or at the right-most vertex of $X_1$. Since $\min A(i_2) \leq \min A(i_k)$, it means that $[1,x] \cap A(i_2) \neq \emptyset$ and hence $\pi(i_1)=[1,x]$. By Lemma \ref{lem:interval}, since $i_k$ values $X_1$ at least as highly as her maximin fair share at $I'$, $i_{1}$ evaluates $X_1$ at least as valuable as her maximin fair share at $I'$. Now due to the minimality of $x$, we get that $[1,x] \subseteq X_1$. Hence, it follows that $((X_1\setminus X)\cup X_2,X_3,\ldots,X_{|N'|})$ is a partition certifying that $\mms_{i_k}(I'') \ge \mms_{i_k}(I')$. Thus, by the induction hypothesis, every $i_k$ with $k>j$ receives a connected bundle of value at least $\mms_{i_k}(I')$. 
\end{proof}

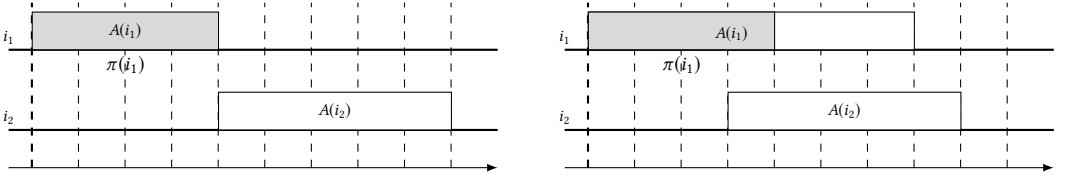
\begin{figure*}
\begin{tabular}{cc}
\begin{minipage}{0.45\textwidth}
\begin{tikzpicture}[scale=0.62,transform shape]
\draw [-latex](-0.5,0) coordinate(dd)-- (0,0) coordinate (O1) -- (10,0)coordinate(ff) node[above]{};
\draw [dashed,thick] (O1) -- (0,0.8) coordinate(B) -- (0,2.5) coordinate(A) -- ++(0,1)coordinate(ff2);

\draw [thick] (dd|-A) node[above]{$i_{1}$}-- (A-|ff);
\draw [thick] (dd|-B) node[above]{$i_{2}$}-- (B-|ff);

\foreach \xx in{1,2,...,9}{
\draw[dashed] (\xx,0) -- (\xx,0|- ff2);
}

\begin{scope}[shift={(A)}]
\node[above right=0.4cm and 0cm of A,right,draw, minimum width=4cm,minimum height=0.8cm,fill=gray!30](n3b) {$A(i_1)$};
\node[below =0cm of n3b](n3c) {\Large $\pi(i_1)$};
\end{scope}

\begin{scope}[shift={(B)}]
\coordinate(O2) at (4,0);
\node[above right=0.4cm and 0cm of O2,right,draw, minimum width=5cm,minimum height=0.8cm,fill=white](n4a) {$A(i_{2})$};
\end{scope}
\end{tikzpicture}
\end{minipage}
\qquad\quad
\begin{minipage}{0.45\textwidth}
\begin{tikzpicture}[scale=0.62,transform shape]
\draw [-latex](-0.5,0) coordinate(dd)-- (0,0) coordinate (O1) -- (10,0)coordinate(ff) node[above]{};
\draw [dashed,thick] (O1) -- (0,0.8) coordinate(B) -- (0,2.5) coordinate(A) -- ++(0,1)coordinate(ff2);

\draw [thick] (dd|-A) node[above]{$i_{1}$}-- (A-|ff);
\draw [thick] (dd|-B) node[above]{$i_{2}$}-- (B-|ff);

\foreach \xx in{1,2,...,9}{
\draw[dashed] (\xx,0) -- (\xx,0|- ff2);
}

\begin{scope}[shift={(A)}]
\node[above right=0.4cm and 0cm of A,right,draw, minimum width=7cm,minimum height=0.8cm,fill=white](n3a) {};
\node[above right=0.4cm and 0cm of A,right,draw, minimum width=4cm,minimum height=0.8cm,fill=gray!30](n3b) {};
\node[below =0cm of n3b](n3c) {\Large $\pi(i_1)$};
\node[above right = 0.1cm of n3c] {$A(i_1)$};
\end{scope}

\begin{scope}[shift={(B)}]
\coordinate(O2) at (3,0);
\node[above right=0.4cm and 0cm of O2,right,draw, minimum width=5cm,minimum height=0.8cm,fill=white](n4a) {$A(i_{2})$};
\end{scope}
\end{tikzpicture}
\end{minipage}
\end{tabular}
\caption{Illustration of the algorithm in the proof of Theorem \ref{thm:PO:MMS:approval}.}\label{fig:nonnested:alg}
\end{figure*}

We give an example where the algorithm used in Theorem~\ref{thm:PO:MMS:approval} fails to find a PO and MMS allocation when agents' preferences are given by approval intervals that do not obey the non-nestedness condition.
\begin{example}[Example where the moving-knife algorithm fails for nested intervals]
\label{ex:nested}
\upshape
Consider an instance with two agents and five vertices on a path, with binary additive valuations as below.

\begin{center}
	\upshape
	\setlength{\tabcolsep}{6.4pt}
	\begin{tabular}{rccccc}
		\toprule
		&
		\multicolumn{5}{l}{\!\!\!
			\begin{tikzpicture}[scale=0.51, transform shape, every node/.style={minimum size=8mm, inner sep=1.2pt, font=\huge}]	
			\node[draw, circle](2)  at (1.2,0) {$v_1$};
			\node[draw, circle](3)  at (2.4,0) {$v_2$};
			\node[draw, circle](4)  at (3.6,0) {$v_3$};
			\node[draw, circle](5)  at (4.8,0) {$v_4$};
			\node[draw, circle](6)  at (6,0) {$v_5$};
			\draw[-, >=latex,thick] (2)--(3) (3)--(4) (4)--(5) (6)--(5);
			\end{tikzpicture}\!\!\!\!
			\vspace{-2pt}
		} \\
		\midrule
		Alice:\!\! &  1 &  1 & 1 & 1 & 1 \\
		Bob:\!\! & 0 & 1 & 1 & 0 & 0 \\
		\bottomrule
\end{tabular}
\end{center}
\smallskip

Applying the algorithm in Theorem~\ref{thm:PO:MMS:approval}, we give $\{v_1,v_2\}$ to Alice and $\{v_3,v_4,v_5\}$ to Bob. This allocation is Pareto-dominated by the allocation giving $\{v_1,v_2,v_3\}$ to Bob and $\{v_4,v_5\}$ to Alice. Noticeably, this example does not admit an MMS allocation that maximizes the utilitarian social welfare: the unique utilitarian optimal allocation is the allocation giving everything to Alice, which clearly violates the MMS requirement for Bob.
\qed
\end{example}

\section{Conclusion}
In this work, we have studied the computational complexity of finding Pareto-efficient outcomes, in the natural setting where we need to allocate indivisible items into connected bundles. We showed that although finding a Pareto-optimal allocation is easy for some topologies, this does not extend to general trees. Further, we proved that when imposing additional fairness requirements, finding a Pareto-optimum becomes NP-hard even when the underlying item graph is a path. We have also seen that a Pareto-optimal EF1 allocation may not exist with the contiguity requirement while such an allocation always exists when these requirements are ignored. 

While we have focused on the divisions of goods, studying an allocation of {\em chores} with graph-connectivity constraints is an interesting future direction. In particular, one may ask what graph structures give positive results in terms of both existence and computational complexity. Finally, several recent papers studied the fair division problem over social networks \citep{abebe,Bei2017,Gourves2017,Bredereck2018} where a social network describes the envy relation between agents. A particular focus is laid on {\em local envy-freeness}, requiring that each agent does not envy the bundle of her neighbours. Although our graph describes a relationship among {\em items} rather than agents, it would be interesting to analyze `intermediate' cases. For example, suppose that we only focus on the envy between a pair of agents who are allocated adjacent bundles in $G$, what graph structure guarantees the existence of a locally envy-free and Pareto-optimal allocation? 

\smallskip
\noindent
\textbf{Acknowledgements.}
We thank reviewers at AAAI-19 and at AI$^3$ for helpful feedback.
This work was supported by ERC grant 639945 (ACCORD).

%
%



\end{document}